\documentclass[reprint,aps,prl,superscriptaddress,longbibliography,nofootinbib]{revtex4-2}

\usepackage{amsmath,amssymb,amsthm,mathtools}
\usepackage{braket}
\usepackage{xcolor}
\usepackage{graphicx}
\usepackage{hyperref}
\hypersetup{colorlinks=true,linkcolor=blue!60!black,citecolor=blue!60!black,urlcolor=blue!60!black}
\usepackage{bm}
\usepackage{orcidlink}

\usepackage{booktabs}
\usepackage{siunitx}
\usepackage{quantikz}

\newtheorem{theorem}{Theorem}
\newtheorem{lemma}[theorem]{Lemma}
\newtheorem{proposition}[theorem]{Proposition}
\newtheorem{corollary}[theorem]{Corollary}
\theoremstyle{definition}
\newtheorem{example}[theorem]{Example}

\newcommand{\Id}{\mathbf{1}}
\newcommand{\II}{\mathcal I}
\newcommand{\Tr}{\operatorname{Tr}}
\newcommand{\PhiP}{\Phi^{+}}
\newcommand{\Dn}{\mathcal{D}_{\eta}}
\newcommand{\Mn}{\mathcal{M}}
\newcommand{\Rep}{\mathcal R}
\newcommand{\EAC}{\mathrm{EAC}}

\begin{document}

\title{Full-Rank Noise Forbids Long-Range Entanglement Swapping}

\author{Xuan Du Trinh\,\orcidlink{0009-0009-5610-462X}}
\email{xtrinh@cs.stonybrook.edu}
\affiliation{Stony Brook University, Stony Brook, New York 11794, USA}

\author{Nicholas Pardave}
\email{pardan@rpi.edu}
\affiliation{Department of Computer Science, Rensselaer Polytechnic Institute, Troy, New York 12180, USA}

\author{Angie Huang}
\email{angiehuang@college.harvard.edu}
\affiliation{Harvard University, Cambridge, Massachusetts 02138, USA}

\author{Xiangyi Meng\,\orcidlink{0000-0001-5184-7648}}
\email{xmenggroup@gmail.com}
\affiliation{Department of Physics, Applied Physics, and Astronomy, Rensselaer Polytechnic Institute, Troy, New York 12180, USA}%
\affiliation{Network Science and Technology Center, Rensselaer Polytechnic Institute, Troy, New York 12180, USA}

\author{Nengkun Yu}
\email{nengkun.yu@cs.stonybrook.edu}
\affiliation{Stony Brook University, Stony Brook, New York 11794, USA}

\date{\today}

\begin{abstract}
Quantum repeaters extend entanglement by swapping noisy elementary
links. We prove that full-rank noise forbids this at long range:
for any entangled full-rank two-qubit link, there is a finite depth
beyond which no end-to-end entanglement can be established
regardless of the measurement outcomes on intermediate qubits, even
under any adaptive postselected strategy. This limit is set by
one spectral parameter of the link, giving a no-go
criterion for repeater routing. In contrast, we construct link state families of rank three and
rank two that admit postselected measurement outcome branches of exponentially small
probability but with strictly positive concurrence at every finite
depth. Swapping experiments on a superconducting processor show
that links of equal initial concurrence but different rank behave
differently under postselected swapping. In the language of many-body physics, the chain is a
matrix-product density operator, and full-rank bonds forbid
long-range localizable entanglement, while rank-deficient bonds
can sustain it.

\end{abstract}

\maketitle

Entanglement shared between distant nodes is a central resource for the
quantum internet: it enables secure communication, distributed
quantum computing, and entanglement-enhanced measurements
between remote sensors~\cite{Kimble2008,WehnerElkoussHanson2018}. Because the
performance of direct quantum communication through noisy channels
deteriorates exponentially with distance, quantum repeaters seek to
establish long-range entanglement by dividing the channel into shorter
elementary links and connecting them through repeated entanglement
swapping~\cite{briegel1998quantum,dur1999quantum}. If every
elementary link were a perfect Bell pair, each swap would reproduce a
Bell pair on the outer qubits, and the chain could be extended
indefinitely~\cite{Swapping-original-Ekert,PanSwapping1998}. Real links, however, are
noisy. Each swap propagates the residual entanglement and
accumulates the imperfections of the constituent links, so the
end-to-end state deteriorates with chain depth. The central issue is
not merely how rapidly this entanglement decays, but whether it can
remain nonzero at every finite depth or instead disappear after a
finite number of swaps.

The distinction between gradual decay and separability at a
finite depth is operational: weakly
entangled pairs may be distilled, but once the shared state
becomes separable, there is no entanglement left to
recover~\cite{HHH1997distill}. For isotropic depolarized links under standard
Bell-measurement swapping, the effective depolarizing parameters
multiply along the chain~\cite{dur1999quantum,SenDe2005swapping}, so
the end-to-end state becomes separable at a finite depth.
For arbitrary noisy links, a single swap has been characterized
in closed form~\cite{PhysRevA.94.012336}. Both studies consider the
standard Bell-measurement protocol.

We consider the broader class of adaptive postselected
swapping protocols: each
intermediate node in the chain jointly measures its two qubits
(Fig.~\ref{fig:chain}) and may postselect any pure entangled
outcome. The outcome need not be
maximally entangled as in the standard protocol, and its
success probability may be arbitrarily small.
Each outcome record defines a postselected measurement outcome
branch. The measurements may proceed sequentially, with later choices
adapted to earlier results, or in parallel, with the accepted branch
selected after the complete set of classical records is known. The
strategy may pursue any objective.

Three observations raise our research question. First, the
finite-depth separability reviewed above is established only for
the isotropic depolarizing model. In practice, the
noise model of a real elementary link is rarely
known exactly. Noise from the environment causes
errors in every direction, without being isotropic. Full-rank links are therefore more
generic and arguably more realistic. Second, the link can also be of lower rank, and at the lowest
rank, pure entangled links clearly can keep entanglement alive at
every depth~\cite{Swapping-original-Ekert,PhysRevA.94.012336}. Third, postselection by local filtering can increase the
entanglement of a single link~\cite{Verstraete2001}, so adaptive
postselected protocols have much more freedom than the standard
protocol to deliver an entangled end-to-end state. As a natural
consequence, we ask:
\emph{can an adaptive postselected swapping protocol establish
end-to-end entanglement at arbitrarily large depth when the
elementary link is of full rank, and how does the answer change at
lower ranks?}

We answer the question completely. For full-rank links, the
answer is \emph{no}: for
every such entangled elementary link, there is a finite depth
beyond which the end-to-end state is separable for any
adaptive postselected swapping protocol (Theorem~\ref{thm:rank4}). We call the smallest such depth
the entanglement horizon of the link. The theorem gives an explicit upper bound on
the horizon, depending only on the smallest eigenvalue of the link's
density matrix, and the bound is uniform over all postselected
branches. The bound also applies beyond entanglement swapping to every
finite-round protocol of stochastic local operations and classical
communication (SLOCC) that uses one copy of each link, allowing the
nodes to perform local operations, exchange classical messages, and
postselect outcomes (Corollary~\ref{cor:slocc-horizon}).

For lower ranks, the answer is \emph{yes}: low-rank entangled link families can
support a postselected branch that has exponentially small
probability, but the end-to-end entanglement survives at every
finite depth.
We consider the qubit amplitude damping and pure dephasing
channels, together with their composition. Applied to one qubit of the Bell state $\PhiP$, they generate families where the composition gives links of rank three,
and either channel alone gives links of rank two. We prove that for these rank-deficient links, the swapping chain that
postselects the $\Phi^\pm$ branches produces an entangled
end-to-end state at every finite depth, with total accepted probability $(1/2)^N$
(Example~\ref{ex:lowrank}).

\begin{figure}[t]
\centering
\includegraphics[width=0.9\columnwidth,page=1]{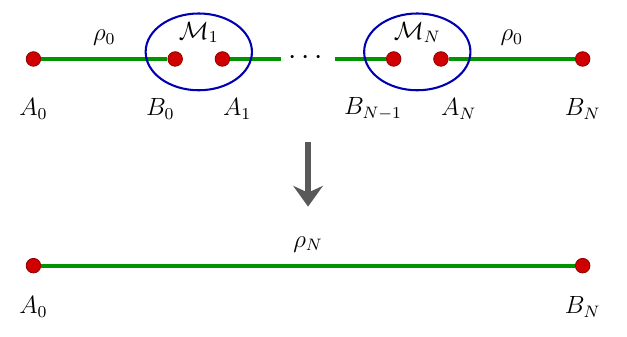}
\caption{Postselected swapping chain. The $N+1$ links $\rho_0$
occupy the qubit pairs $(A_i,B_i)$, and intermediate node $i$
postselects the entangled outcome
$\ket{\Psi_i}=(\II\otimes M_i)\ket{\PhiP}$ on its pair
$(B_{i-1},A_i)$. Tracing out the measured qubits leaves the state
$\rho_N$ of the boundary pair $A_0B_N$, in which
each outcome enters as the filter map
$\mathcal{M}_i(X)=M_i^*XM_i^T$. The outcomes $\ket{\Psi_i}$ need not be
Bell states or maximally entangled. Postselection keeps only
the branch with the intended record $\{M_i\}$, and the rule choosing
each measurement may adapt to the outcomes already observed.}
\label{fig:chain}
\end{figure}

Together, the results trace the contrast between full-rank and
rank-deficient links to a single mechanism. In a full-rank
link, the smallest eigenvalue supplies a noise component that
accumulates along the chain and eventually makes every postselected
output separable. In a rank-deficient link, no such component
exists, and its absence opens the possibility of postselecting a
rare branch that survives at every depth. The existence of an entanglement horizon is therefore
sensitive to the support of the link and is not determined by the
amount of entanglement alone: one link may have a finite horizon
while another of equal concurrence evades it. We also
test the rank contrast on a superconducting processor, where links
of nearly equal initial concurrence follow different decay
trajectories.

\paragraph*{Setup.}
Consider a chain of $N+1$ copies of an elementary link $\rho_0$
shared by $N+2$ nodes, where the $N$ intermediate nodes perform the
swapping measurements and the two end nodes hold the boundary
qubits $A_0$ and $B_N$ (Fig.~\ref{fig:chain}). Every two-qubit link
can be written as the Choi operator of a single-qubit completely
positive (CP) map,
$\rho_0=(\II\otimes\mathcal{E})(\PhiP)$, where
$\PhiP:=\ket{\PhiP}\!\bra{\PhiP}$ is the Bell-state
projector~\cite{Jamiolkowski1972,Choi1975}. For a postselected pure
entangled swap outcome
$\ket{\Psi_i}=(\II\otimes M_i)\ket{\PhiP}$ at node $i$,
normalized by
$\Tr(M_i^\dagger M_i)=2$, define the CP filter map
$\mathcal{M}_i(X):=M_i^*XM_i^T$. The unnormalized postselected
end-to-end state $\widetilde\rho_N$ on $A_0B_N$ is then
proportional to the Choi operator of the composite CP map, as proved in
Sec.~\ref{sm:onestep} of the Supplemental Material
(SM):
\begin{equation}\label{eq:chain}
\widetilde\rho_N \;\propto\;
(\II\otimes\mathcal{N})(\PhiP),\qquad
\mathcal{N}
=
\mathcal{E}\circ\mathcal{M}_N\circ\mathcal{E}\circ\cdots
\circ\mathcal{M}_1\circ\mathcal{E}.
\end{equation}

\paragraph*{Reduction to the isotropic chain.}
Full rank means that the smallest eigenvalue $\mu_{\min}(\rho_0)$ is
strictly positive.
Write $\lambda:=4\mu_{\min}(\rho_0)\in(0,1)$
($\lambda<1$, since $\lambda=1$ would make $\rho_0=I_4/4$
separable).
Then $\lambda$ is the largest weight of the maximally mixed
state $I_4/4$ that can be subtracted from $\rho_0$ while leaving a
positive operator in the rest, and peeling it off gives
$\rho_0=\lambda\,I_4/4+(1-\lambda)\sigma$.
The link therefore carries a noise floor $\lambda$.
The state $\sigma$ must be entangled to keep $\rho_0$
entangled. Wootters' equal-concurrence
ensemble~\cite{wootters1998entanglement} decomposes
$\sigma=\sum_i p_i\ket{\psi_i}\!\bra{\psi_i}$, in which every
Wootters piece $\ket{\psi_i}$ carries the same concurrence
$C(\psi_i)=C(\sigma)>0$. Folding the noise back into each term gives
$$
\rho_0=\sum_i p_i\,\tau_{\psi_i},
\qquad
\tau_{\psi_i}:=(1-\lambda)\ket{\psi_i}\!\bra{\psi_i}+\tfrac{\lambda}{4}I_4.
$$

Consider the isotropic state whose depolarizing parameter
$\eta$ is set by the noise floor $\lambda$:
\begin{equation}\label{eq:eta-of-lambda}
\omega_\eta:=\eta\,\PhiP+(1-\eta)\frac{I_4}{4},
\qquad
\eta:=\frac{2(1-\lambda)}{2-\lambda}\in(0,1).
\end{equation}
In SM Sec.~\ref{sm:slocc}, we construct a CP map
$\mathcal E_{\psi_i}$ with explicit product Kraus operators such
that $\tau_{\psi_i}=\mathcal E_{\psi_i}(\omega_\eta)$. Since the Wootters pieces are pure states of equal
concurrence, the $\tau_{\psi_i}$ are all equivalent under local
unitaries.
We now show that the isotropic chain emerges already at the
level of the swapping chain's full-rank input. Write
$\vec i=(i_0,\dots,i_N)$ for the Wootters piece indices of the
$N+1$ links and $p_{\vec i}:=\prod_kp_{i_k}$. Then
$$
\rho_0^{\otimes(N+1)}
=\sum_{\vec i}p_{\vec i}
\bigotimes_{k=0}^{N}\tau_{\psi_{i_k}}
=
\sum_{\vec i}p_{\vec i}\,
\Bigl(\bigotimes_{k=0}^{N}\mathcal E_{\psi_{i_k}}\Bigr)
\bigl(\omega_\eta^{\otimes(N+1)}\bigr).
$$
Each summand on the right is the same isotropic chain input
$\omega_\eta^{\otimes(N+1)}$ processed by the maps
$\mathcal E_{\psi_{i_k}}$ on the individual links. The chain
of postselected swaps is linear in the joint input
$\rho_0^{\otimes(N+1)}$, so it acts on the sum
summand by summand: if no summand yields an entangled end-to-end
output, the full output state is separable. The maps $\mathcal E_{\psi_{i_k}}$ have a special
construction (SM Sec.~\ref{sm:slocc}): the Kraus operators are
product, the first one acts on qubit $A_k$ only, while the remaining
ones annihilate $\PhiP$. This construction splits the full output state
$\widetilde\rho_N$ of Eq.~\eqref{eq:chain} into two parts,
$\widetilde\rho_N=\widetilde\rho_N^{(0)}+\widetilde\rho_N^{(1)}$,
where $\widetilde\rho_N^{(1)}$ is separable and
$\widetilde\rho_N^{(0)}$ is a nonnegative combination of
isotropic-chain outputs, each of the form of
Fig.~\ref{fig:chain} run on $N+1$ copies of $\omega_\eta$ (with a
different branch of measurement outcomes $\{M_i\}$ for
each, SM Secs.~\ref{sm:slocc}--\ref{sm:caseA}). Therefore, to prove that $\widetilde\rho_N$ is separable at
large enough depth, it suffices to prove it for each isotropic-chain
output: the full-rank problem reduces to the isotropic chain. We
proceed with the following mechanism.

\paragraph*{Mechanism: weak entanglement killed by local depolarization.}
The isotropic-link channel is the qubit depolarizer
$\Dn=\eta\,\mathrm{id}+(1-\eta)\Rep$: with probability $\eta$ it
transmits the qubit untouched, and with probability $1-\eta$ the
replacement $\Rep(X)=\Tr(X)I_2/2$ resets it to the maximally mixed
state, so that $\omega_\eta=(\II\otimes\Dn)(\PhiP)$. The isotropic chain's output is given by Eq.~\eqref{eq:chain} with the
link channel $\mathcal E=\Dn$. Denote its end-to-end state by
$\widetilde\rho_N^{\mathrm{iso}}$ and its composite map by
$$
\mathcal N_{\mathrm{iso}}
=\Dn\circ\Mn_N\circ\Dn\circ\cdots\circ\Mn_1\circ\Dn.
$$
Expanding $\mathcal N_{\mathrm{iso}}$ into elementary terms
and collecting them (SM Secs.~\ref{sm:lastR}--\ref{sm:product})
gives the decomposition
\begin{equation}\label{eq:state-grouped}
\widetilde\rho_N^{\mathrm{iso}}\;\propto\;
(\II\otimes\mathcal N_{\mathrm{iso}})(\PhiP)
=\eta^{N+1}P_K+S_N,
\end{equation}
where $P_K=(\II\otimes K)\PhiP(\II\otimes K)^\dagger$ is the pure
component filtered by the composite swap
$K=M_N^*\cdots M_1^*$, and $S_N$ is an unnormalized separable
component. The end-to-end state is separable iff its partial transpose
is positive semidefinite, by the Peres--Horodecki
criterion~\cite{Peres1996,Horodecki1996}. The partial
transpose of the pure term $\eta^{N+1}P_K$ has exactly one
negative eigenvalue, $-\eta^{N+1}\Gamma_K/2$, with
$\Gamma_K:=|\det K|=\prod_i|\det M_i|\le 1$
(SM Sec.~\ref{sm:bellblock}). The separable component $S_N$
collects product terms with positive partial transposes, whose
weights along the corresponding eigenvector may compensate
that negativity. We prove that once
$\eta^{N+1}\Gamma_K/2$ is small enough, these terms
compensate the negative eigenvalue completely, part by part,
each up to its entanglement absorption capacity in the
language of Ref.~\cite{eac-paper}, and the whole
partial transpose becomes positive semidefinite. Specifically, the end-to-end state
is separable whenever $\eta^{N+1}\le 1/3$, uniformly over every
realized record $\{M_i\}$ (SM Sec.~\ref{sm:cancellation}), so
the bound covers every branch of any sequential, parallel, or
adaptive strategy.
The interpretation of the mechanism is intuitive: in the
decomposition of the end-to-end state, Eq.~\eqref{eq:state-grouped},
the depolarizers supply the product terms when $N$ increases, while
the entangled pure component shrinks exponentially as $\eta^{N+1}$.
When the only pure entangled component is weak enough, the
end-to-end state becomes separable.

\begin{theorem}[Full-rank entanglement horizon]\label{thm:rank4}
For every entangled full-rank two-qubit elementary link $\rho_0$,
the postselected end-to-end state of the swapping chain of $N+1$
copies is separable for every realized record of pure entangled
intermediate outcomes whenever
\begin{equation}\label{eq:horizon}
\eta^{\,N+1}\ \le\ \frac{1}{3},
\end{equation}
where $\eta\in(0,1)$ is given by Eq.~\eqref{eq:eta-of-lambda}.
Consequently, under any adaptive postselected swapping protocol,
a finite entanglement horizon exists.
\end{theorem}

\begin{corollary}[Single-copy SLOCC horizon]
\label{cor:slocc-horizon}
On the same chain of $N{+}1$ links, every nonzero-probability
end-to-end state of a finite-round single-copy SLOCC protocol at
the nodes, not limited to entanglement swapping, is separable
whenever condition~\eqref{eq:horizon} holds. The model and proof are in SM
Sec.~\ref{sm:slocc-reduction}.
\end{corollary}

Because $\eta<1$, condition~\eqref{eq:horizon} is met at the
finite depth $N+1\ge\log 3/\log(1/\eta)$ and remains true at every
greater depth. This depth therefore gives an explicit branch-independent
upper bound on the entanglement horizon.
For example, a link with noise floor $\lambda=10^{-2}$ cannot
carry entanglement beyond roughly $220$ swaps. For a particular link or branch,
the end-to-end state may become separable earlier.
Theorem~\ref{thm:rank4} and Corollary~\ref{cor:slocc-horizon}
extend to heterogeneous full-rank links
$\{\rho_i\}_{i=0}^{N}$: with $\lambda_i:=4\mu_{\min}(\rho_i)$ and
$\eta_i:=2(1-\lambda_i)/(2-\lambda_i)$, the same conclusions
hold whenever $\prod_{i=0}^{N}\eta_i\le 1/3$.
This branch-independent entanglement horizon bound is a property of the
noisy links and is not specific to a particular protocol.
Purification does not open an escape route from the horizon. Any
finite-round SLOCC protocol on finitely many copies of a full-rank
link returns a pair that is separable or again of full rank
(SM Sec.~\ref{sm:slocc-reduction}), so a chain of purified links
again obeys Theorem~\ref{thm:rank4}. For that reason, purification may
reduce the noise floor and extend the horizon, but cannot remove it.

This settles the full-rank case. Noisy links with lower rank,
however, do not always force a finite horizon. We now exhibit low-rank
entangled link families for which the swapping chains admit
postselected branches whose end-to-end states are entangled at
every finite depth.

\begin{example}[Low-rank persistence]
\label{ex:lowrank}
The amplitude-damping channel
$\mathcal A_s(X)=K_0XK_0^\dagger+K_1XK_1^\dagger$
describes a process in which a qubit in $\ket1$ survives with
probability $s$ and decays to $\ket0$ with probability $1-s$.
Its Kraus operators are
$K_0=\left(\begin{smallmatrix}1&0\\0&\sqrt{s}\end{smallmatrix}\right)$
and
$K_1=\left(\begin{smallmatrix}0&\sqrt{1-s}\\0&0\end{smallmatrix}\right)$.
The pure-dephasing channel
$\mathcal P_\kappa(X):=\frac{1+\kappa}{2}X+
\frac{1-\kappa}{2}\sigma_zX\sigma_z$
describes a process in which the qubit is unchanged with probability
$(1+\kappa)/2$ and undergoes a phase flip $\sigma_z$ with probability
$(1-\kappa)/2$. It preserves the populations and shrinks the
off-diagonal entries by a factor $\kappa\in(0,1]$. Their
composition
$\mathcal E_{s,\kappa}:=\mathcal P_\kappa\circ\mathcal A_s
=\mathcal A_s\circ\mathcal P_\kappa$
reduces to amplitude damping at $\kappa=1$ and to pure dephasing at
$s=1$.
Sending one qubit of $\PhiP$ through $\mathcal E_{s,\kappa}$, we
have, in the computational basis, the elementary link
\begin{equation}\label{eq:ad-link}
\rho_{s,\kappa}:=(\II\otimes\mathcal E_{s,\kappa})(\PhiP)
=\frac{1}{2}\begin{pmatrix}
1 & 0 & 0 & \kappa\sqrt{s}\\
0 & 0 & 0 & 0\\
0 & 0 & 1-s & 0\\
\kappa\sqrt{s} & 0 & 0 & s
\end{pmatrix}.
\end{equation}%
The corresponding concurrence is
$C(\rho_{s,\kappa})=\kappa\sqrt{s}>0$.
The rank of the link is three under the composition ($0<s<1$ and
$0<\kappa<1$), and we have a rank-two link under amplitude
damping alone ($\kappa=1$, $0<s<1$) or pure dephasing alone ($0<\kappa<1$, $s=1$), and
a rank-one link at $s=\kappa=1$, where
$\rho_{s,\kappa}=\PhiP$.
In a chain of $N+1$ copies of $\rho_{s,\kappa}$, we postselect
every intermediate Bell measurement on $\PhiP$ or $\Phi^-$. The
two outcomes together carry probability $1/2$ at each node, so the
accepted branches carry $(1/2)^N$ in total. Applying a local
Pauli-$Z$ on the boundary qubit for each $\Phi^-$ outcome, we obtain, for
every realized record, the postselected end-to-end state
$\rho_N=\rho_{s^{N+1},\kappa^{N+1}}$, again a member of the same family
(SM Sec.~\ref{sm:lowrank}). Therefore,
$C(\rho_N)=(\kappa\sqrt{s})^{N+1}>0$ for every finite $N$. Each accepted
branch of the swapping process is exponentially rare, but on it
the entanglement of the end-to-end state never vanishes. This is what Theorem~\ref{thm:rank4}
forbids for full-rank links.
\end{example}

\paragraph{Experimental Results.}
To test this fundamental rank-aware distinction, we performed swapping experiments on IBM's Heron R3 quantum processor, comparing state evolution under the rank-2 amplitude-damping (AD) noise versus rank-4 depolarizing (DP) noise. While both implementations incur non-negligible overhead due to hardware imperfections, we mitigated this by analyzing their concurrence ratio, $C(\text{DP})/C(\text{AD})$, assuming the device overhead is multiplicative, noise-independent, and cancels out in the ratio (SM Sec.~\ref{sm:repswaps}). As shown in Fig.~\ref{fig:rank-comparison}, starting from an initial concurrence ratio $C(\text{DP})/C(\text{AD})\approx 1$, the ratio monotonically decays to zero over the number of swaps, in agreement with the prediction that $C(\text{DP})$ vanishes strictly before $C(\text{AD})$.
This highlights the fundamental distinction between full-rank and rank-deficient two-qubit states. Eventually, accumulated experimental noise elevates the rank-2 AD noise to full rank and drives both $C(\text{DP})$ and $C(\text{AD})$ to zero (Fig.~\ref{fig:rank-comparison}, inset).

\begin{figure}[t]
\centering
\includegraphics[width=\columnwidth]{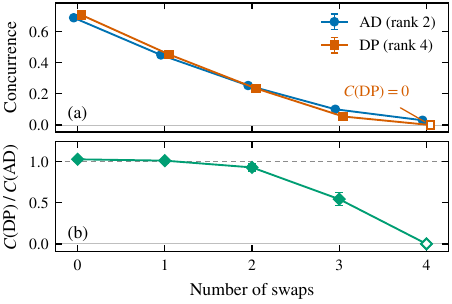}
\caption{Experimental verification of rank-aware distinction in entanglement swapping.
The states are prepared under one-sided, rank-2 amplitude-damping (AD) and rank-4 depolarizing (DP) channels respectively on IBM's Heron R3 processor (\texttt{ibm\_boston}), with an initial concurrence $C(\text{AD})\approx C(\text{DP})\approx 0.7$.
The ratio monotonically decays to zero (solid), signaling the sudden death of $C(\text{DP})$ as the theory predicts (dashed). Eventually both $C(\text{DP})$ and $C(\text{AD})$ reach zero due to experimental imperfection (inset).
}
\label{fig:rank-comparison}
\end{figure}

\paragraph*{Discussion.}
Theorem~\ref{thm:rank4} and Example~\ref{ex:lowrank} together
cover all four ranks: a full-rank link always has a finite
entanglement horizon, whereas at each lower rank, there exist links
that evade it. In its heterogeneous form, $\prod_i\eta_i\le 1/3$, the theorem
is an operational criterion: each link enters the criterion only through its noise
floor $\lambda_i$, one spectral number, and no
other knowledge of the noise model is required.
For qubit chains under SLOCC, Ref.~\cite{MasajadaPhilipStreltsov2026}
bounds, exponentially in the depth, the probability that a
postselected branch has concurrence at least a fixed positive value,
which leaves rare
entangled branches possible at every finite depth.
Corollary~\ref{cor:slocc-horizon} excludes such branches for
full-rank links beyond an explicit finite depth, while the low-rank
example shows that they exist.

\begin{samepage}
In a quantum network, when we consider a route $P$ between two end
nodes to perform a chain of postselected swaps, the criterion
becomes a threshold on an additive score along the route. Assign to each
full-rank link $i$ the normalized cost
\begin{equation}\label{eq:beta-link}
\beta_i:=\frac{1}{\ln3}
\ln\frac{2-\lambda_i}{2(1-\lambda_i)},
\qquad \lambda_i:=4\mu_{\min}(\rho_i),
\end{equation}%
\end{samepage}%
and define the score of the route as
$\beta(P):=\sum_{i\in P}\beta_i$. The horizon bound gives the
branch-independent rejection rule
\begin{equation}\label{eq:budget}
\boxed{\;
\beta(P)\ge1\ \Longrightarrow\
\substack{\text{every postselected}\\[-2pt]
\text{branch is separable}}\; .}
\end{equation}
A simple memory model makes the score concrete: for $N{+}1$
identical isotropic links with initial noise floor $\lambda_0$,
each qubit $q$ stored for a time $t_q$ in a memory with
depolarizing time $T_m$, the score is approximated, when $\lambda_0\ll1$
and $t_q/T_m\ll1$, by
$\beta(P)\approx[(N{+}1)\lambda_0+\sum_q t_q/T_m]/(2\ln3)$
(SM Sec.~\ref{sm:budget}). The first term
is the cost already present when the links are created, and the
second is accumulated during storage. A qubit pair suffers entanglement degradation while waiting for
the remaining links to be generated and for the swap record to
travel to the end nodes.

The horizon bound also sorts protocol architectures into research
directions. Entanglement
purification~\cite{Bennett1996purification,Deutsch1996privacy,briegel1998quantum,dur1999quantum}
cannot revive copies whose state is already
separable. Purification scheduled
after the swaps therefore recovers nothing beyond the horizon
bound. It must instead be applied before the
swaps, or at intermediate stations~\cite{Lapeyre2012fullrank}
spaced, at the least, so that no
purely swapped segment crosses $\beta=1$, and the optimal schedule
under realistic memories is an important problem. At the network layer,
hop-by-hop transmission forwards quantum information
between neighboring nodes without creating end-to-end entanglement,
thereby avoiding the horizon and relaxing the strict synchronization required
by end-to-end designs~\cite{ZengVBQT2026}. With entanglement
verification at each hop, the protocol improves fidelity and
throughput under noise. Determining the optimal strategy that adapts these purification,
hop-by-hop, and hybrid designs is a natural challenge, whose
solution could have a broad impact on quantum network design.

In the language of many-body physics, a chain of noisy entangled pairs shared
by neighboring sites forms a matrix-product density
operator~\cite{VerstraeteGarciaRipollCirac2004}, and the
shared pairs are called its bonds. In that literature,
one measures all the sites between the two ends and
asks how much entanglement the outcomes concentrate between the
ends, and calls that quantity the localizable
entanglement~\cite{VerstraetePoppCirac2004,PoppVerstraeteMartinDelgadoCirac2005}.
Adaptive postselected swapping is one such procedure. For pure
bonds, the localizable entanglement can stay nonzero at every
distance~\cite{WahlPerezGarciaCirac2012}. When the bonds are mixed,
we show that full-rank bonds force every branch to become separable
after a finite number of bonds, while some rank-deficient bonds admit
a postselection scheme that keeps the localizable entanglement
nonzero at every finite distance.
What separates the two behaviors is the rank of the bonds,
not the amount of entanglement they carry.

\medskip
\noindent
\begin{acknowledgments}
X.T. and N.Y. are supported by the National Science Foundation under Grant No.~CCF-2553759. N.P. and X.M. were supported primarily by the Co-design Center for Quantum Advantage (C2QA) through the U.S. Department of Energy under Award No.~DE-SCL0000120. X.M. was also supported in part by the Rensselaer-IBM Accelerate Quantum Computing Research.
\end{acknowledgments}

\bibliographystyle{apsrev4-2}
\bibliography{refs}

\clearpage

\appendix
\onecolumngrid
\begin{center}
{\large\textbf{Supplemental Material for}}\\[3pt]
{\large\textbf{``Full-Rank Noise Forbids Long-Range Entanglement
Swapping''}}\\[8pt]
Xuan Du Trinh,$^{1}$ Nicholas Pardave,$^{2}$ Angie Huang,$^{3}$ Xiangyi Meng,$^{2}$ and Nengkun Yu$^{1}$\\[4pt]
\emph{$^{1}$Stony Brook University, Stony Brook, New York 11794, USA}\\
\emph{$^{2}$Rensselaer Polytechnic Institute, Troy, New York 12180, USA}\\
\emph{$^{3}$Harvard University, Cambridge, Massachusetts 02138, USA}\\[6pt]
\end{center}
\setcounter{section}{0}
\setcounter{equation}{0}
\renewcommand{\thesection}{S\arabic{section}}
\renewcommand{\theequation}{S\arabic{equation}}
\renewcommand{\thetheorem}{S\arabic{theorem}}
\setcounter{theorem}{0}
\setcounter{secnumdepth}{3}
\makeatletter
\@removefromreset{equation}{section}
\makeatother
\renewcommand{\theHtheorem}{S\arabic{theorem}}

\noindent\emph{Roadmap.}
This Supplemental Material is organized in three parts. \emph{Part A}
(Secs.~\ref{sm:onestep}--\ref{sm:slocc-reduction}) proves the
isotropic-chain horizon theorem, from which the
body's full-rank Theorem~\ref{thm:rank4} follows by the
reduction below. The proof proceeds in three stages: a one-step
channel-composition identity (Sec.~\ref{sm:onestep}), the reduction
of any entangled full-rank link to a chain of universal isotropic
links joined by dressed rank-two filters
(Secs.~\ref{sm:slocc}--\ref{sm:caseA}), and a direct partial-transpose positivity argument that yields
the sharp threshold $\eta^{N+1}\le 1/3$ uniformly over the
postselected outcomes (Secs.~\ref{sm:lastR}--\ref{sm:cancellation}).
Sec.~\ref{sm:budget}
derives the normalized route score in main-text
Eq.~\eqref{eq:budget} and controls its low-noise approximation.
Sec.~\ref{sm:slocc-reduction} then proves the single-copy
SLOCC corollary and the purification statement.
\emph{Part B} (Sec.~\ref{sm:lowrank}) treats links of rank three
and links of rank two as one channel family and proves in a single
theorem that chains of both deliver an entangled state at
every finite depth under postselected $\Phi^\pm$ swaps.
\emph{Part C}
(Secs.~\ref{sm:repswaps}--\ref{sm:qubits}) describes the methods and devices used in the quantum processor experiments.

\subsection*{Part A: Proof of the isotropic-chain horizon theorem}

\section{One-step effective channel}\label{sm:onestep}

We begin by recording the channel form of one swap step, on which
the entire reduction relies.

\begin{proposition}[One-step effective channel]\label{prop:onestep}
Two links $\rho_{A_1B_1}=(\II\otimes\mathcal E_1)(\PhiP)$ and
$\rho_{A_2B_2}=(\II\otimes\mathcal E_2)(\PhiP)$, joined by the
postselected outcome $\ket{\Psi}=(\II\otimes M)\ket{\PhiP}$ on
$B_1A_2$, produce on $A_1B_2$ the unnormalized state
\begin{equation}\label{eq:onestep}
\widetilde\rho_{A_1B_2}\;\propto\;(\II\otimes\mathcal E_2\circ\mathcal M
\circ\mathcal E_1)(\PhiP),
\qquad \mathcal M(X)=M^*XM^T.
\end{equation}
\end{proposition}

\begin{proof}
The input state is $\rho_{A_1B_1}\otimes\rho_{A_2B_2}$, and the
retained outcome is the rank-one projector $\Psi:=\ket\Psi\!\bra\Psi$ on
$B_1A_2$. We apply the measurement operator $\II\otimes\Psi\otimes\II$ on
$A_1B_1A_2B_2$ and trace out the measured qubits, which leaves on $A_1B_2$
the unnormalized state
\begin{equation}\label{eq:onestep-trace}
\widetilde\rho_{A_1B_2}
=\Tr_{B_1A_2}\bigl[(\II\otimes\Psi\otimes\II)
(\rho_{A_1B_1}\otimes\rho_{A_2B_2})\bigr].
\end{equation}
Here, $\Psi^2=\Psi$ and the cyclicity of the partial trace over the
qubits on which $\Psi$ acts absorb the projector on the right of the
projected state
$(\II\otimes\Psi\otimes\II)(\rho_{A_1B_1}\otimes\rho_{A_2B_2})
(\II\otimes\Psi\otimes\II)$. We simplify
Eq.~\eqref{eq:onestep-trace} with three identities. First, since $\Psi$
is rank one, the partial trace is a contraction,
\begin{equation}\label{eq:onestep-contract}
\widetilde\rho_{A_1B_2}
=(\II\otimes\bra\Psi\otimes\II)
(\rho_{A_1B_1}\otimes\rho_{A_2B_2})
(\II\otimes\ket\Psi\otimes\II),
\end{equation}
as one sees by evaluating the partial trace in a basis $\{\ket n\}$ of
$B_1A_2$ and using $\sum_n\ket n\!\braket{n|\Psi}=\ket\Psi$. Second,
Kraus decompositions $\mathcal E_1(X)=\sum_jK^{(1)}_jXK^{(1)\dagger}_j$
and $\mathcal E_2(X)=\sum_\ell K^{(2)}_\ell XK^{(2)\dagger}_\ell$ write
the two links as the sums
$\rho_{A_1B_1}=\sum_j(\II\otimes K^{(1)}_j)\PhiP(\II\otimes K^{(1)}_j)^\dagger$
and
$\rho_{A_2B_2}=\sum_\ell(\II\otimes K^{(2)}_\ell)\PhiP(\II\otimes K^{(2)}_\ell)^\dagger$.
Third, for every $j$ and $\ell$,
\begin{equation}\label{eq:onestep-vector}
(\II\otimes\bra\Psi\otimes\II)
\bigl[(\II\otimes K^{(1)}_j)\ket{\PhiP}_{A_1B_1}
\otimes(\II\otimes K^{(2)}_\ell)\ket{\PhiP}_{A_2B_2}\bigr]
=\tfrac12(\II\otimes K^{(2)}_\ell M^*K^{(1)}_j)\ket{\PhiP}_{A_1B_2}.
\end{equation}
To prove Eq.~\eqref{eq:onestep-vector}, expand
$\ket{\PhiP}=\tfrac1{\sqrt2}\sum_i\ket{ii}$ in all three places:
$(\II\otimes K^{(1)}_j)\ket{\PhiP}_{A_1B_1}
=\tfrac1{\sqrt2}\sum_i\ket i_{A_1}\otimes K^{(1)}_j\ket i_{B_1}$,
$(\II\otimes K^{(2)}_\ell)\ket{\PhiP}_{A_2B_2}
=\tfrac1{\sqrt2}\sum_k\ket k_{A_2}\otimes K^{(2)}_\ell\ket k_{B_2}$, and
$\bra\Psi=\bra{\PhiP}(\II\otimes M^\dagger)
=\tfrac1{\sqrt2}\sum_m\bra m_{B_1}\otimes\bra m_{A_2}M^\dagger$.
Contracting $B_1$ and $A_2$ leaves
$\tfrac1{2\sqrt2}\sum_{i,k}
\bigl(\sum_m\bra mK^{(1)}_j\ket i\,\bra mM^\dagger\ket k\bigr)
\ket i_{A_1}\otimes K^{(2)}_\ell\ket k_{B_2}$. The inner sum is
$(K^{(1)\mathsf T}_jM^\dagger)_{ik}$, and
$\sum_k(K^{(1)\mathsf T}_jM^\dagger)_{ik}\ket k
=(K^{(1)\mathsf T}_jM^\dagger)^{\mathsf T}\ket i=M^*K^{(1)}_j\ket i$, so
the contraction equals
$\tfrac1{2\sqrt2}\sum_i\ket i_{A_1}\otimes K^{(2)}_\ell M^*K^{(1)}_j\ket i_{B_2}$,
the right side of Eq.~\eqref{eq:onestep-vector}. Inserting the Kraus
sums into Eq.~\eqref{eq:onestep-contract} and applying
Eq.~\eqref{eq:onestep-vector} to each term gives
\begin{equation}\label{eq:onestep-result}
\widetilde\rho_{A_1B_2}
=\tfrac14\sum_{j,\ell}(\II\otimes K^{(2)}_\ell M^*K^{(1)}_j)\PhiP
(\II\otimes K^{(2)}_\ell M^*K^{(1)}_j)^\dagger
=\tfrac14(\II\otimes\mathcal E_2\circ\mathcal M\circ\mathcal E_1)(\PhiP),
\end{equation}
which is Eq.~\eqref{eq:onestep}.
\end{proof}

\medskip
\noindent\textbf{Rank-one filters are entanglement breaking.}
If the swap matrix $M$ has rank one, write $M=\ket u\!\bra v$.
Then $(\II\otimes M)\rho(\II\otimes M^\dagger)=\sigma_A\otimes\ket u\!\bra u$
is a product positive operator for every $\rho$, so the boundary
state is separable already after one step. Throughout the main text, we restrict to rank-two $M_i$, the only rank that can sustain
entanglement through the swaps. This rank is that of the filter
$M_i$ itself, and is unrelated to the rank of the link state it acts
on.

\section{Product-Kraus representation}\label{sm:slocc}

We now realize each component $\tau_\psi$ as the image of the
universal isotropic state $\omega_\eta$ under a completely
positive map whose Kraus operators are product operators on
$A\otimes B$. Recall the
Wootters decomposition of $\rho_0$ used in the main text,
\begin{equation}\label{eq:rho0-decomp}
\rho_0=\sum_i p_i\,\tau_{\psi_i},
\qquad
\tau_\psi:=(1-\lambda)\ket\psi\!\bra\psi+\tfrac{\lambda}{4}I_4,
\end{equation}
with uniform noise weight $\lambda/4=\mu_{\min}(\rho_0)$ across the
pieces (Wootters' equal-concurrence ensemble~\cite{wootters1998entanglement}).
The universal isotropic reference state is
$\omega_\eta=\eta\,\PhiP+(1-\eta)I_4/4=(\II\otimes\Dn)(\PhiP)$ with
the depolarizing parameter of Eq.~\eqref{eq:eta-of-lambda}.

\begin{lemma}[Product-Kraus representation]\label{lem:slocc}
Every two-qubit pure state $\ket\psi$ determines a unique
$2\times2$ matrix $M_\psi$ with $\Tr(M_\psi^\dagger M_\psi)=1$ via
$(M_\psi\otimes I)\ket{\PhiP}=\ket\psi/\sqrt2$ (distinct from the
swap-filter normalization $\Tr(M_i^\dagger M_i)=2$ fixed in the
body). For each such $M_\psi$, there are vectors
$\xi_0,\xi_1\in\mathbb C^2$ with
$\xi_0\xi_0^\dagger+\xi_1\xi_1^\dagger=(2-\lambda)(I-M_\psi M_\psi^\dagger)$.
We write them as column vectors rather than as kets because they are
not normalized: taking the trace gives
$\xi_0^\dagger\xi_0+\xi_1^\dagger\xi_1=2-\lambda$.
The completely positive map $\mathcal E_\psi$ with the Kraus
operators
\begin{equation}\label{eq:kraus-types}
K_0=\sqrt{2-\lambda}\,M_\psi\otimes I,\qquad
L_{ab}=(\xi_a\bra0)\otimes(\ket b\!\bra1),\quad a,b\in\{0,1\},
\end{equation}
then satisfies $\mathcal E_\psi(\omega_\eta)=\tau_\psi$
when $\eta$ is given by Eq.~\eqref{eq:eta-of-lambda}. Each $L_{ab}$
is a tensor product of two single-qubit operators, and
$L_{ab}\ket{\PhiP}=0$.
\end{lemma}

\begin{proof}
In the computational basis, $(M_\psi\otimes I)\ket{\PhiP}=\ket\psi/\sqrt2$
is equivalent to $(M_\psi)_{ij}=\langle ij|\psi\rangle$, so $M_\psi$
exists, is unique, and has
$\Tr(M_\psi^\dagger M_\psi)=\braket{\psi|\psi}=1$. The operator
$K_0=\sqrt{2-\lambda}\,M_\psi\otimes I$ satisfies
$K_0\PhiP K_0^\dagger=\tfrac{2-\lambda}{2}\ket\psi\!\bra\psi$ and
$K_0K_0^\dagger=(2-\lambda)\,(M_\psi M_\psi^\dagger)\otimes I$. Since
$M_\psi M_\psi^\dagger\succeq0$ has trace one,
$(2-\lambda)(I-M_\psi M_\psi^\dagger)\succeq0$. We take $\xi_0$ and
$\xi_1$ to be its two orthonormal eigenvectors, and we scale each so
that $\xi_a^\dagger\xi_a$ equals its eigenvalue, which gives
$\xi_0\xi_0^\dagger+\xi_1\xi_1^\dagger=(2-\lambda)(I-M_\psi M_\psi^\dagger)$.
The operators
$L_{ab}$ of Eq.~\eqref{eq:kraus-types} are tensor products of the
single-qubit operators $\xi_a\bra0$ and $\ket b\!\bra1$. Each
annihilates $\ket{\PhiP}$, because $\bra0\otimes\bra1$ is orthogonal
to it, and
$\sum_{ab}L_{ab}L_{ab}^\dagger
=\sum_a\xi_a\xi_a^\dagger\otimes\sum_b\ket b\!\bra b
=(2-\lambda)(I-M_\psi M_\psi^\dagger)\otimes I$. Hence, the map
$\mathcal E_\psi(\rho)=K_0\rho K_0^\dagger+\sum_{ab}L_{ab}\rho L_{ab}^\dagger$
satisfies $\mathcal E_\psi(\PhiP)=\tfrac{2-\lambda}{2}\ket\psi\!\bra\psi$ and
$\mathcal E_\psi(I_4)=K_0K_0^\dagger+(2-\lambda)(I-M_\psi M_\psi^\dagger)\otimes I
=(2-\lambda)\,I_4$, so
$$
\mathcal E_\psi(\omega_\eta)
=\tfrac{(2-\lambda)\eta}{2}\ket\psi\!\bra\psi
+\tfrac{(2-\lambda)(1-\eta)}{4}\,I_4
=(1-\lambda)\ket\psi\!\bra\psi+\tfrac\lambda4\,I_4=\tau_\psi,
$$
where the second equality uses $(2-\lambda)\eta=2(1-\lambda)$ and
$(2-\lambda)(1-\eta)=\lambda$, both from Eq.~\eqref{eq:eta-of-lambda}.
No other value of $\eta$ serves: the normalized image
$\mathcal E_\psi(\omega_\eta)/\Tr\mathcal E_\psi(\omega_\eta)$ has pure
weight $\eta/(2-\eta)$, which equals $1-\lambda$ only at this value.
Only complete positivity, the product form of the $L_{ab}$, and
$L_{ab}\ket{\PhiP}=0$ enter the chain expansion below.
\end{proof}

\section{Filter-moving identity}\label{sm:filtermove}

Lemma~\ref{lem:slocc} puts the one-sided operator
$K_0=\sqrt{2-\lambda}\,M_\psi\otimes I$ on the qubit $A_k$ of every link $k$.
At a node $k\ge1$, the swap measures $A_k$ together with $B_{k-1}$, and
the next lemma shows that the swap filter $M_k$ can absorb any operator
on $A_k$ without changing the postselected output. Recall from the
main text that node $k$ retains the outcome
$\ket{\Psi_k}=(\II\otimes M_k)\ket{\PhiP}$ on the pair $(B_{k-1},A_k)$,
where $M_k$ acts on $A_k$.

\begin{lemma}[Filter-moving identity]\label{lem:filtermove}
Let $k\ge1$, let $X$ be a single-qubit operator, and let $X_{A_k}$
act as $X$ on $A_k$ and as the identity on every other qubit. For every
operator $\rho$ on the chain qubits,
\begin{equation}\label{eq:filtermove}
\Tr_{B_{k-1}A_k}\!\bigl[\ket{\Psi_k}\!\bra{\Psi_k}\,X_{A_k}\rho X_{A_k}^\dagger\bigr]
=\Tr_{B_{k-1}A_k}\!\bigl[\ket{\Psi_k'}\!\bra{\Psi_k'}\,\rho\bigr],
\qquad
\ket{\Psi_k'}:=(\II\otimes M_k')\ket{\PhiP},
\quad
M_k':=X^\dagger M_k.
\end{equation}
Applying $X$ to $A_k$ before the swap at node $k$ therefore gives
the same unnormalized output as the swap with the dressed filter
$M_k'$ and no $X$.
\end{lemma}

\begin{proof}
The operator $X_{A_k}$ acts only on qubits that the partial trace
removes, and the partial trace over $B_{k-1}A_k$ is cyclic for such
operators, so
$$
\Tr_{B_{k-1}A_k}\!\bigl[\ket{\Psi_k}\!\bra{\Psi_k}\,X_{A_k}\rho X_{A_k}^\dagger\bigr]
=\Tr_{B_{k-1}A_k}\!\bigl[X_{A_k}^\dagger\ket{\Psi_k}\!\bra{\Psi_k}X_{A_k}\,\rho\bigr].
$$
Since $X^\dagger$ and $M_k$ both act on $A_k$,
$X_{A_k}^\dagger\ket{\Psi_k}=(\II\otimes X^\dagger M_k)\ket{\PhiP}=\ket{\Psi_k'}$,
and $\bra{\Psi_k}X_{A_k}$ is its adjoint $\bra{\Psi_k'}$, which gives
Eq.~\eqref{eq:filtermove}.
\end{proof}

Two remarks complete the picture. First, the dressed filter
$M_k'=X^\dagger M_k$ need not satisfy the normalization
$\Tr[(M_k')^\dagger M_k']=2$ that the main text imposes on swap
filters. Rescaling $M_k'$ by a positive number multiplies both sides of
Eq.~\eqref{eq:filtermove} by its square and leaves the normalized
postselected state unchanged, so we restore the normalization whenever
we need it, and separability is unaffected. Second, the lemma needs a
swap at node $k$, so it does not apply to link $0$. At link $0$, the
operator $K_0$ acts on the boundary qubit $A_0$, which no node measures.
Every swap projector and every partial trace acts on other qubits and
commutes with $K_0$, so $K_0$ survives as a one-sided operation on
$A_0$ applied after the whole chain, which is the outer dressing in
Eq.~\eqref{eq:reduction} below. A local operation on one boundary qubit
maps separable states to separable states, so it does not affect the
separability argument.

\section{Assembling the reduction and separability of the product-Kraus terms}\label{sm:caseA}

We now combine the product-Kraus representation (Lemma~\ref{lem:slocc}) and
the filter-moving identity (Lemma~\ref{lem:filtermove}) into an explicit
decomposition of the full-rank chain into a piece built from locally dressed isotropic chains (the only piece that can carry end-to-end entanglement) and a piece that is separable by construction.

Multilinearity of the chain map in its link inputs lets us expand $\rho_0^{\otimes(N+1)}$ via the Wootters decomposition in Eq.~\eqref{eq:rho0-decomp} and expand each $\tau_{\psi_{i_k}}$ through the Kraus decomposition in Eq.~\eqref{eq:kraus-types}. Each resulting term selects,
at every link $k$, either the one-sided Kraus operator $K_0$ or one of the product Kraus operators $L_{ab}$. Grouping the terms by whether \emph{any} link selects an $L_{ab}$, we write the chain state as
\begin{align}\label{eq:reduction}
\widetilde\rho_N
&=\;\underbrace{\widetilde\rho_N^{(0)}}_{\text{all-}K_0\text{ terms}}
\;+\;
\underbrace{\widetilde\rho_N^{(1)}}_{\substack{\text{terms with at}\\\text{least one }L_{ab}}},\\
\widetilde\rho_N^{(0)}
&=(2-\lambda)^{N+1}\sum_{\vec i} p_{\vec i}\,
(M_{\psi_{i_0}}\!\otimes\!I)\,
\widetilde\rho_N^{\mathrm{iso}}\!\bigl(\{M_k'(\vec i)\}\bigr)\,
(M_{\psi_{i_0}}^\dagger\!\otimes\!I),
\notag
\end{align}
where $\vec i=(i_0,\dots,i_N)$ indexes the Wootters components on
the $N+1$ links, $p_{\vec i}=\prod_k p_{i_k}$ is the product of
the probabilities of the selected equal-concurrence pieces in
Eq.~\eqref{eq:rho0-decomp}, and
$M_k'(\vec i):=M_{\psi_{i_k}}^\dagger M_k$ for $k\ge 1$. Here
$\widetilde\rho_N^{\mathrm{iso}}(\{M_k'\})$ denotes the end-to-end
state of an isotropic chain of $N+1$ links $\omega_\eta$ joined by
the intermediate filters $\{M_k'\}$.
The second line collects the all-$K_0$ terms. The $N+1$ factors
$K_0\omega_\eta K_0^\dagger$ contribute $(2-\lambda)^{N+1}$ and leave
$\omega_\eta$ on every link with $M_{\psi_{i_k}}$ on $A_k$,
Lemma~\ref{lem:filtermove} absorbs $M_{\psi_{i_k}}$ into the swap at
node $k\ge1$, and $M_{\psi_{i_0}}$ on the boundary qubit $A_0$ commutes
with every swap and survives as the outer dressing.
For one swap, $N=1$, the all-$K_0$ term with indices $(i_0,i_1)$ reads
\begin{align}
&(2-\lambda)^2\,\Tr_{B_0A_1}\!\bigl[\ket{\Psi_1}\!\bra{\Psi_1}\,
(M_{\psi_{i_0}}\!\otimes I\otimes M_{\psi_{i_1}}\!\otimes I)\,
(\omega_\eta\otimes\omega_\eta)\,
(M_{\psi_{i_0}}\!\otimes I\otimes M_{\psi_{i_1}}\!\otimes I)^\dagger\bigr]\notag\\
&\qquad=(2-\lambda)^2\,(M_{\psi_{i_0}}\!\otimes I)\,
\Tr_{B_0A_1}\!\bigl[\ket{\Psi_1'}\!\bra{\Psi_1'}\,
(\omega_\eta\otimes\omega_\eta)\bigr]\,
(M_{\psi_{i_0}}\!\otimes I)^\dagger,
\qquad
\ket{\Psi_1'}=(\II\otimes M_{\psi_{i_1}}^\dagger M_1)\ket{\PhiP}.
\label{eq:reduction-example}
\end{align}
Lemma~\ref{lem:filtermove} absorbs the factor $M_{\psi_{i_1}}$ on the
measured qubit $A_1$ into the swap, and the factor $M_{\psi_{i_0}}$ on
the unmeasured qubit $A_0$ commutes with the projector and the trace
and moves outside. The trace on the right is
$\widetilde\rho_1^{\mathrm{iso}}(M_1')$, and summing over $(i_0,i_1)$
with the weight $p_{i_0}p_{i_1}$ gives the second line of
Eq.~\eqref{eq:reduction} at $N=1$.
Each $M_k'=M_{\psi_{i_k}}^\dagger M_k$ is invertible. The factor $M_{\psi_{i_k}}$ is invertible because every Wootters piece is entangled and the concurrence of a two-qubit pure state satisfies $C(\psi)=2|\det M_\psi|$, so $\det M_{\psi_{i_k}}\neq0$. The factor $M_k$ is invertible by the rank-two swap assumption. When we invoke the isotropic theorem below, we rescale each $M_k'$ to satisfy $\Tr[(M_k')^\dagger M_k']=2$, which changes only the overall
postselection probability and not separability.

We handle the two groups differently.
Every term of the all-$K_0$ sector $\widetilde\rho_N^{(0)}$ is an
isotropic-chain output dressed by a local operator on $A_0$ and
weighted by $(2-\lambda)^{N+1}p_{\vec i}\ge0$, so the sector is
separable as soon as every isotropic chain
$\widetilde\rho_N^{\mathrm{iso}}(\{M_k'(\vec i)\})$ is. This sector is the
only one that can carry end-to-end entanglement, and the one
the mechanism of Secs.~\ref{sm:lastR}--\ref{sm:cancellation}
analyzes.
The complementary sector $\widetilde\rho_N^{(1)}$ collects every term with at least one product Kraus operator $L_{ab}$ and is
separable on the boundary $A_0\otimes B_N$, as proved in the
remainder of this section.

\begin{quote}\itshape
\textbf{Claim (separability of the product-Kraus terms).}
In the Kraus expansion based on Eq.~\eqref{eq:kraus-types},
whenever any link $k\in\{0,\ldots,N\}$ selects a product Kraus
operator $L_{ab}$ of Eq.~\eqref{eq:kraus-types},
the resulting term in the chain-state sum is separable on the
boundary qubits $A_0\otimes B_N$.
\end{quote}

\begin{proof}
Each term of the expansion applies the chain map, the swap projections
followed by the trace over the measured qubits, to a tensor product of
$N+1$ link operators, each of which is $K_0\omega_\eta K_0^\dagger$ or
one of the $L_{ab}\omega_\eta L_{ab}^\dagger$. Fix a term in which link
$k$ selects $L_{ab}$.
By Eq.~\eqref{eq:kraus-types}, $L_{ab}=(\xi_a\bra0)\otimes(\ket b\!\bra1)$
is a tensor product of single-qubit operators, and $L_{ab}\ket{\PhiP}=0$.
The operator $L_{ab}$ therefore annihilates the $\PhiP$ component of
$\omega_\eta$, so the link operator of link $k$ is
$$
L_{ab}\omega_\eta L_{ab}^\dagger=\tfrac{1-\eta}{4}\,L_{ab}L_{ab}^\dagger
=P_k\otimes Q_k,
\qquad
P_k:=\tfrac{1-\eta}{4}\,\xi_a\xi_a^\dagger\ \text{on }A_k,
\quad
Q_k:=\ket b\!\bra b\ \text{on }B_k .
$$
The tensor product of the link operators is therefore a product
across the cut between $A_k$ and $B_k$. Every swap acts on a pair
$(B_{j-1},A_j)$, and every such pair lies on one side of this cut, as
does every measured qubit. The boundary operator is therefore a tensor
product of an operator on $A_0$ and an operator on $B_N$, which is
separable.
\end{proof}

\section{Last-replacement grouping}\label{sm:lastR}

The remaining all-$K_0$ sector is, by construction, a chain of
$N{+}1$ isotropic links $\omega_\eta$ joined by dressed rank-two
filters. We expand this chain by the link of the last
replacement. Throughout
Secs.~\ref{sm:lastR}--\ref{sm:cancellation}, $\mathcal N$ denotes
the composite map of this isotropic chain, written
$\mathcal N_{\mathrm{iso}}$ in the main text. Define the replacement map $\Rep(X)=\Tr(X)\,I_2/2$, so
$\Dn=\eta\,\mathrm{id}+(1-\eta)\Rep$. The grouping below is the precise form of the
main-text intuition: each depolarizer either multiplies the intact
trajectory by $\eta$ or, with weight $1-\eta$, replaces the qubit and
starts a separable-noise trajectory.
In $\mathcal N=\Dn\circ\Mn_N\circ\Dn\circ\cdots\circ\Mn_1\circ\Dn$,
the depolarizer of link $0$ is applied first, and the depolarizer of
link $k\ge1$ is the one that follows $\Mn_k$.

Set $\mathcal F:=\Mn_N\circ\cdots\circ\Mn_1$ (filters only),
$\mathcal L_k:=\Mn_k\circ\Dn\circ\cdots\circ\Mn_1\circ\Dn$ (the
part of $\mathcal N$ applied before the depolarizer of link $k$,
with $\mathcal L_0:=\mathrm{id}$), and
$\mathcal Q_k:=\Mn_N\circ\cdots\circ\Mn_{k+1}$ (the filters applied
after it, with $\mathcal Q_N:=\mathrm{id}$).

\begin{lemma}[Last-replacement expansion]\label{lem:lastR}
The composite map $\mathcal N$ of the isotropic chain satisfies
\begin{equation}\label{eq:lastR-sm}
\mathcal N = \eta^{N+1}\,\mathcal F
+\sum_{k=0}^{N}\eta^{N-k}(1-\eta)\,\mathcal Q_k\circ\Rep\circ\mathcal L_k.
\end{equation}
\end{lemma}

\begin{proof}
We prove Eq.~\eqref{eq:lastR-sm} by induction on $N$. For $N=0$,
$\mathcal N=\Dn=\eta\,\mathrm{id}+(1-\eta)\Rep$, which is
Eq.~\eqref{eq:lastR-sm} with
$\mathcal F=\mathcal L_0=\mathcal Q_0=\mathrm{id}$. For $N\ge1$, write
$\mathcal N=\Dn\circ\Mn_N\circ\mathcal N'$, where $\mathcal N'$ is the
composite map of the chain of links $0,\dots,N-1$, and expand the
outermost depolarizer,
$$
\mathcal N=\eta\,\Mn_N\circ\mathcal N'+(1-\eta)\,\Rep\circ\Mn_N\circ\mathcal N'.
$$
The induction hypothesis gives
$\mathcal N'=\eta^{N}\mathcal F'
+\sum_{k=0}^{N-1}\eta^{N-1-k}(1-\eta)\,\mathcal Q_k'\circ\Rep\circ\mathcal L_k$
with $\mathcal F':=\Mn_{N-1}\circ\cdots\circ\Mn_1$ and
$\mathcal Q_k':=\Mn_{N-1}\circ\cdots\circ\Mn_{k+1}$. Since
$\Mn_N\circ\mathcal F'=\mathcal F$ and $\Mn_N\circ\mathcal Q_k'=\mathcal Q_k$,
the first term is
$\eta^{N+1}\mathcal F+\sum_{k=0}^{N-1}\eta^{N-k}(1-\eta)\,\mathcal Q_k\circ\Rep\circ\mathcal L_k$.
Since $\Mn_N\circ\mathcal N'=\mathcal L_N$ and $\mathcal Q_N=\mathrm{id}$,
the second term is the $k=N$ summand. Together they are
Eq.~\eqref{eq:lastR-sm}.
\end{proof}

The $k$th summand of Eq.~\eqref{eq:lastR-sm} collects the trajectories
in which the depolarizer of link $k$ replaces the qubit and every later
depolarizer leaves it intact. These are the trajectories whose last
replacement is at link $k$, and the first term is the trajectory without
any replacement.

\section{Product structure of the noise terms}\label{sm:product}

Each replacement trajectory in Eq.~\eqref{eq:lastR-sm} yields a product
positive operator on the boundary. For any operator $\sigma_{AB}$,
$(\II\otimes\Rep)(\sigma_{AB})=\Tr_B(\sigma_{AB})\otimes I_2/2$, since
both sides are linear and agree on every
$\ket i\!\bra j\otimes\ket k\!\bra l$. Applying the $k$th term to $\PhiP$ factor by factor and using that $\mathcal Q_k$ acts on $B$ alone, we obtain
\begin{equation}\label{eq:final-product}
(\II\otimes\mathcal Q_k\circ\Rep\circ\mathcal L_k)(\PhiP)
=G_k^{(0)}\otimes H_k^{(0)},
\qquad
G_k^{(0)}:=\Tr_B\bigl[(\II\otimes\mathcal L_k)(\PhiP)\bigr],
\qquad
H_k^{(0)}:=\mathcal Q_k(I_2/2).
\end{equation}
Both factors are positive, because $\mathcal L_k$ and $\mathcal Q_k$ are
compositions of completely positive maps, and
$G_0^{(0)}=H_N^{(0)}=I_2/2$ at the two ends. The index $k$ labels the
trajectory, not the link qubits $A_k,B_k$, and the superscript $(0)$
marks the marginals before the alignment of Sec.~\ref{sm:bellblock}.
Substituting into Eq.~\eqref{eq:lastR-sm} gives the unnormalized
decomposition
\begin{equation}\label{eq:state-grouped-sm}
X_N:=(\II\otimes\mathcal N)(\PhiP)
=\eta^{N+1}P_K
+\sum_{k=0}^{N}w_k\,G_k^{(0)}\otimes H_k^{(0)},
\qquad
w_k:=\eta^{N-k}(1-\eta),
\end{equation}
where $P_K=(\II\otimes K)\PhiP(\II\otimes K)^\dagger$ and
$K=M_N^*\cdots M_1^*$. The first term is the trajectory in which no
replacement fired, and the sum is the separable term $S_N$ of main-text
Eq.~\eqref{eq:state-grouped}.

\section{Bell-block inequality and pure-component alignment}\label{sm:bellblock}

The refined decomposition of Eq.~\eqref{eq:state-grouped-sm} separates the
chain into one pure component and a sum of product noise terms. The
pure component is
$$
P_K=(\II\otimes K)\PhiP(\II\otimes K)^\dagger,\qquad
K=M_N^*\cdots M_1^* .
$$
Let $K=URV^\dagger$ be its singular value decomposition (SVD) with
$R=\mathrm{diag}(\alpha,\beta)$, $\alpha\ge\beta>0$, and set
$\Gamma_K:=|\det K|=\alpha\beta$. Expanding
$\ket{\PhiP}=\tfrac1{\sqrt2}\sum_i\ket{ii}$ gives
$(\II\otimes V^\dagger)\ket{\PhiP}=(V^*\otimes\II)\ket{\PhiP}$, so
$(\II\otimes K)\ket{\PhiP}=(V^*\otimes U)(\II\otimes R)\ket{\PhiP}$ and
\begin{equation}\label{eq:pure-aligned}
P_K=(V^*\otimes U)\,P_R\,(V^*\otimes U)^\dagger,
\qquad
P_R:=(\II\otimes R)\PhiP(\II\otimes R)
=\tfrac12\bigl(\alpha\ket{00}+\beta\ket{11}\bigr)
\bigl(\alpha\bra{00}+\beta\bra{11}\bigr).
\end{equation}
The partial transpose of $P_R$ has the eigenvalues $\alpha^2/2$,
$\beta^2/2$, $\alpha\beta/2$, and $-\alpha\beta/2$. The partial
transpose of $P_K$ is $(V^*\otimes U^*)P_R^{T_B}(V^*\otimes U^*)^\dagger$,
which has the same spectrum, so it has a single negative eigenvalue,
$-\Gamma_K/2$. Multiplicativity of the determinant along
$K=M_N^*\cdots M_1^*$ gives
$\Gamma_K=\prod_{i=1}^N|\det M_i|$. Each factor obeys
$|\det M_i|\le 1$: the singular values $s_{i,1},s_{i,2}$ of $M_i$
satisfy $s_{i,1}^2+s_{i,2}^2=\Tr(M_i^\dagger M_i)=2$, so
$|\det M_i|=s_{i,1}s_{i,2}\le\tfrac12(s_{i,1}^2+s_{i,2}^2)=1$,
with equality iff $s_{i,1}=s_{i,2}=1$, that is, iff $M_i$ is
unitary. Hence, $\Gamma_K\le 1$, with equality iff every $M_i$ is
unitary, as stated in the main text.

The following Bell-block inequality controls the product noise
terms. Let
$$
Q:=(\PhiP)^{T_B}.
$$

\begin{lemma}[Bell-block inequality]\label{lem:bellblock}
For any $A,B\succeq 0$ on $\mathbb C^2$,
$$
A\otimes B^T+2\sqrt{\det A\det B}\,Q\succeq 0.
$$
The constant is sharp.
\end{lemma}

The sharp constant $2\sqrt{\det A\det B}$ is the entanglement
absorption capacity $\EAC(A\otimes B)$ of product noise in the
Bell frame, in the language of Ref.~\cite{eac-paper}, where the
inequality is proved in more detail. We include the short derivation
needed here to keep the present argument self-contained.

\begin{proof}
By continuity, it suffices to prove the claim for positive definite
$A,B$. Since $Q=F/2$, where $F$ is the swap operator and
$(U\otimes U)F(U^\dagger\otimes U^\dagger)=F$, a joint unitary
congruence lets us diagonalize $A$ without changing the statement:
conjugating by $U\otimes U$ fixes $Q$ and preserves both
determinants, while sending $B^T$ to $UB^TU^\dagger$, again positive
definite. Choosing $U$ to diagonalize $A$, we may therefore take $A$
diagonal and $B$ a general positive definite matrix. Write
$$
A=\begin{pmatrix}a&0\\0&d\end{pmatrix},
\qquad
B=\begin{pmatrix}b&z\\ z^*&c\end{pmatrix},
\qquad
\Delta:=\det B=bc-|z|^2,
$$
and set $t:=\sqrt{\det A\det B}=\sqrt{ad\,\Delta}$. In the permuted
basis $(\ket{00},\ket{11},\ket{01},\ket{10})$,
$A\otimes B^T+2tQ=A\otimes B^T+tF$ is
$$
\begin{pmatrix}
ab+t&0&az^*&0\\
0&dc+t&0&dz\\
az&0&ac&t\\
0&dz^*&t&db
\end{pmatrix}.
$$
The upper-left diagonal block is positive definite. Its Schur complement is
$$
\begin{pmatrix}
\dfrac{a(a\Delta+ct)}{ab+t}&t\\[6pt]
t&\dfrac{d(d\Delta+bt)}{dc+t}
\end{pmatrix}.
$$
The diagonal entries of the Schur complement are nonnegative, and its
determinant is zero because $t^2=ad\,\Delta$. Hence, the Schur
complement, and therefore the full matrix, is positive semidefinite.
Sharpness of the constant is shown in Ref.~\cite{eac-paper} and is
not used below.
\end{proof}

\begin{corollary}[Aligned one-block PPT inequality]\label{cor:aligned-bellblock}
Let
$$
\Omega_R:=(\II\otimes R)Q(\II\otimes R).
$$
Then for any $A,B\succeq0$,
\begin{equation}\label{eq:aligned-bellblock}
A\otimes B^T
+\frac{2\sqrt{\det A\,\det B}}{\Gamma_K}\,\Omega_R
\succeq0 .
\end{equation}
\end{corollary}

\begin{proof}
Set $\widetilde B:=R^{-1}BR^{-1}$. Lemma~\ref{lem:bellblock} gives
$$
A\otimes\widetilde B^T
+2\sqrt{\det A\,\det\widetilde B}\,Q\succeq0 .
$$
Conjugating by $\II\otimes R$ preserves positivity and turns the
first term into $A\otimes R\widetilde B^TR=A\otimes B^T$, since $R$ is
real and diagonal. Since
$\det\widetilde B=\det B/(\alpha\beta)^2$, the second term becomes
$$
\frac{2\sqrt{\det A\,\det B}}{\alpha\beta}
(\II\otimes R)Q(\II\otimes R),
$$
which is Eq.~\eqref{eq:aligned-bellblock}, since $\alpha\beta=\Gamma_K$
and $(\II\otimes R)Q(\II\otimes R)=\Omega_R$.
\end{proof}

We now align the chain state with the Schmidt frame of the pure component.
Apply the local unitary
$$
W:=V_A^T\otimes U_B^\dagger .
$$
Since $W=(V^*\otimes U)^\dagger$, Eq.~\eqref{eq:pure-aligned} gives
$WP_KW^\dagger=P_R$, and
$$
P_R^{T_B}=(\II\otimes R)Q(\II\otimes R)=\Omega_R .
$$
The product marginals $G_k^{(0)},H_k^{(0)}$ of
Eq.~\eqref{eq:state-grouped-sm} have aligned versions
$$
G_k:=V^TG_k^{(0)}V^*,\qquad
H_k:=U^\dagger H_k^{(0)}U,
$$
so that $W(G_k^{(0)}\otimes H_k^{(0)})W^\dagger=G_k\otimes H_k$.
These local unitaries preserve determinants, so determinant bounds
proved below for $G_k^{(0)},H_k^{(0)}$ also hold for $G_k,H_k$. Therefore,
the aligned partial transpose of the unnormalized chain state is
\begin{equation}\label{eq:rho-aligned-PT}
(WX_NW^\dagger)^{T_B}
=\eta^{N+1}\Omega_R+\sum_{k=0}^{N}w_k\,G_k\otimes H_k^T .
\end{equation}
Local-unitary alignment does not affect separability, and
normalization by $\Tr X_N$ does not affect PPT positivity. Hence, the
remaining task is only to show that the positive product blocks in
Eq.~\eqref{eq:rho-aligned-PT} compensate the negative direction of
$\eta^{N+1}\Omega_R$.

\section{\texorpdfstring{Determinant bounds on $G_k$ and $H_k$}{Determinant bounds on Gk and Hk}}\label{sm:detbounds}

To apply the aligned Bell-block inequality to the noise terms of
Eq.~\eqref{eq:state-grouped-sm}, we bound the determinants of the
marginals $G_k,H_k$ uniformly over the filters $\{M_i\}$. This is
the technical point that makes the main-text mechanism independent
of the postselected swaps: a filter can reshape the local marginals,
but the depolarizing steps never reduce the determinant lower bound.
Thus, no choice of rank-two outcomes can make the injected separable
noise too weak to compensate the exponentially small pure-component
contribution.

\begin{quote}\itshape
\textbf{Claim.}
Let $G_k^{(0)},H_k^{(0)}$ be the pre-alignment marginals in
Eq.~\eqref{eq:state-grouped-sm}. For every $k=0,\dots,N$ and all
filters $\{M_i\}$,
$\det G_k^{(0)}\ge\tfrac14\prod_{i=1}^k|\det M_i|^2$ and
$\det H_k^{(0)}=\tfrac14\prod_{i=k+1}^N|\det M_i|^2$.
Consequently, $4\sqrt{\det G_k^{(0)}\det H_k^{(0)}}\ge\Gamma_K$. By the
determinant-preservation argument in Sec.~\ref{sm:bellblock}, the same
bounds also hold for the aligned marginals $G_k,H_k$ used in
Eq.~\eqref{eq:rho-aligned-PT}.
\end{quote}

\begin{proof}
For $H_k^{(0)}$, $\det\Mn_i(X)=\det(M_i^*XM_i^T)=|\det M_i|^2\det X$,
so iterating from $\det(I_2/2)=1/4$ gives the stated equality.

For $G_k^{(0)}$, let $\mathcal L_k^\dagger$ denote the Hilbert--Schmidt
adjoint of $\mathcal L_k$, so that
$\Tr[\mathcal L_k(X)]=\Tr[\mathcal L_k^\dagger(I_2)\,X]$.
Expanding $\PhiP=\tfrac12\sum_{ij}\ket{ii}\!\bra{jj}$,
$$
G_k^{(0)}=\tfrac12\sum_{ij}\ket i\!\bra j\,\Tr[\mathcal L_k(\ket i\!\bra j)]
=\tfrac12\,[\mathcal L_k^\dagger(I_2)]^T,
\qquad
\det G_k^{(0)}=\tfrac14\det[\mathcal L_k^\dagger(I_2)] .
$$
The adjoints of the building blocks are $\Mn_i^\dagger(Y)=M_i^TYM_i^*$ and
$\Dn^\dagger=\Dn$, so
$\mathcal L_k^\dagger(I_2)=\Dn(\Mn_1^\dagger(\Dn(\cdots\Mn_k^\dagger(I_2))))$
is obtained from $I_2$ by applying $\Mn_k^\dagger,\Dn,\dots,\Mn_1^\dagger,\Dn$
in turn. The filter step
$Y\mapsto M_i^TYM_i^*$ multiplies the determinant by $|\det M_i|^2$ and
preserves positivity. The depolarizing step $Y\mapsto\Dn(Y)$ never
decreases the determinant of a positive operator, because for $Y\succeq0$
with eigenvalues $y_1,y_2$, the eigenvalues of $\Dn(Y)$ are
$\eta y_i+(1-\eta)(y_1+y_2)/2$, which gives
\begin{equation}\label{eq:depol-det-bound}
\det\Dn(Y)-\det Y=\tfrac{1-\eta^2}{4}(y_1-y_2)^2\ \ge\ 0 .
\end{equation}
Hence, $\det[\mathcal L_k^\dagger(I_2)]\ge\prod_{i=1}^k|\det M_i|^2$, which
is the stated bound. The product of the two results gives
$4\sqrt{\det G_k^{(0)}\det H_k^{(0)}}\ge\prod_{i=1}^N|\det M_i|=\Gamma_K$.
\end{proof}

\section{\texorpdfstring{PPT positivity and the separability threshold}{PPT positivity and the separability threshold}}\label{sm:cancellation}

We now prove Theorem~\ref{thm:iso} in the form used in the main
text. In the decomposition of Eq.~\eqref{eq:state-grouped-sm}, the
partial transpose of the pure component $\eta^{N+1}P_K$ has a single
negative eigenvalue, $-\eta^{N+1}\Gamma_K/2$
(Sec.~\ref{sm:bellblock}), while every remaining term
$w_kG_k^{(0)}\otimes H_k^{(0)}$ is a positive semidefinite product
operator, whose partial transpose is again positive semidefinite. We
show that these product terms compensate the negative eigenvalue
whenever $\eta^{N+1}\le 1/3$, so that the partial transpose of the
whole state is positive semidefinite.

The construction follows the decomposition-certificate technique of
Ref.~\cite{eac-paper}. After a local-unitary alignment, the partial
transpose of the chain state is $\eta^{N+1}\Omega_R$, whose only
negative eigenvalue is $-\eta^{N+1}\Gamma_K/2$, plus a sum of
positive semidefinite product operators. Each product operator
compensates a multiple of $\Omega_R$ up to its entanglement absorption
capacity (Corollary~\ref{cor:aligned-bellblock}), and the weighted
capacities add up to the total multiple of $\Omega_R$ that the sum
compensates. What is specific to the present chain
problem is that the reduction (Secs.~\ref{sm:slocc}--\ref{sm:caseA})
and the uniform determinant bounds (Sec.~\ref{sm:detbounds}) give
every capacity the same lower bound $\tfrac12$, independent of the
link and of the record of rank-two filters. That common lower
bound is what makes the threshold uniform.

\begin{theorem}[Isotropic chain horizon]\label{thm:iso}
Let $\omega_\eta=(\II\otimes\Dn)(\PhiP)$ with $0<\eta<1$, and let
$\widetilde\rho_N^{\mathrm{iso}}$ be the end-to-end state of the chain of $N+1$
copies of $\omega_\eta$ joined by $N$ rank-two intermediate
swap matrices $\{M_i\}$. If $\eta^{N+1}\le 1/3$, then
$\widetilde\rho_N^{\mathrm{iso}}$ is separable, uniformly over the choice of
$\{M_i\}$.
\end{theorem}

\begin{proof}
By Sec.~\ref{sm:bellblock}, after alignment by the local unitary
$W:=V_A^T\otimes U_B^\dagger$ associated with the SVD
$K=M_N^*\cdots M_1^*=URV^\dagger$, the partial transpose of the
unnormalized chain state $X_N$ of Eq.~\eqref{eq:state-grouped-sm} is
\begin{equation}\label{eq:sm-aligned-PT}
(WX_N W^\dagger)^{T_B}
=\eta^{N+1}\,\Omega_R
+\sum_{k=0}^{N}w_k\,G_k\otimes H_k^T ,
\end{equation}
where $\Omega_R=P_R^{T_B}$ has a single negative eigenvalue
$-\Gamma_K/2$ and every $G_k\otimes H_k^T$ is positive semidefinite. For each $k$,
Corollary~\ref{cor:aligned-bellblock} gives
$G_k\otimes H_k^T+c_k\Omega_R\succeq0$, where
$c_k:=2\sqrt{\det G_k\det H_k}/\Gamma_K$ is the entanglement absorption
capacity of the $k$th block in the language of Ref.~\cite{eac-paper}, and
Sec.~\ref{sm:detbounds} gives $c_k\ge\tfrac12$. Multiplying the
$k$th inequality by $w_k\ge0$ and summing over $k$ gives
$$
\sum_{k=0}^{N}w_k\,G_k\otimes H_k^T+S\,\Omega_R\succeq0,
\qquad
S:=\sum_{k=0}^{N}w_kc_k\ \ge\ \tfrac12\sum_{k=0}^{N}w_k
=\frac{1-\eta^{N+1}}{2}.
$$
The coefficient of $\Omega_R$ in Eq.~\eqref{eq:sm-aligned-PT} is
$\eta^{N+1}$, and $\eta^{N+1}\le(1-\eta^{N+1})/2$ if and only if
$\eta^{N+1}\le1/3$. Under this condition
$\theta:=\eta^{N+1}/S\in(0,1]$, and Eq.~\eqref{eq:sm-aligned-PT} can
be rewritten as
$$
(WX_NW^\dagger)^{T_B}
=\theta\Bigl(\sum_{k=0}^{N}w_k\,G_k\otimes H_k^T+S\,\Omega_R\Bigr)
+(1-\theta)\sum_{k=0}^{N}w_k\,G_k\otimes H_k^T ,
$$
a sum of two positive semidefinite operators, so
$(WX_NW^\dagger)^{T_B}\succeq0$. Undoing the alignment conjugates the
partial transpose by the local unitary $V^*\otimes U^*$, and
normalization divides by $\Tr X_N>0$. Both preserve positive
semidefiniteness, so
$(\widetilde\rho_N^{\mathrm{iso}})^{T_B}\succeq0$ and
$\widetilde\rho_N^{\mathrm{iso}}$ is separable by the Peres--Horodecki
criterion~\cite{Peres1996,Horodecki1996}. The only property of the
record used is that each $M_i$ has rank two, which gives
$\Gamma_K>0$, so the threshold $1/3$ is the same for every record
$\{M_i\}$.
\end{proof}

\medskip
\noindent
\emph{The threshold is sharp.} Take every filter unitary, which is
allowed by $\Tr(M_i^\dagger M_i)=2$. The depolarizer $\Dn$ commutes
with unitary conjugation, so the chain channel collapses to
$\mathcal N=\mathcal U\circ\mathcal D_{\eta^{N+1}}$ for a single
unitary conjugation $\mathcal U$, and the end-to-end state is
$\omega_{\eta^{N+1}}$ up to a local unitary. An isotropic state
$\omega_p$ is entangled if and only if $p>1/3$, so the
end-to-end state is entangled for every $\eta^{N+1}>1/3$. No threshold smaller than
$1/3$ can therefore certify separability for all records of
rank-two filters.

\medskip
\noindent
\emph{Consequence.} Combined with the reduction
(Secs.~\ref{sm:slocc}--\ref{sm:caseA}), Theorem~\ref{thm:iso} proves the body's
Theorem~\ref{thm:rank4}: for every entangled full-rank two-qubit
link, the end-to-end state is separable for every realized record
whenever Eq.~\eqref{eq:horizon} holds.
Conditioning covers adaptive protocols. On any realized
branch, the complete outcome record fixes the filters $\{M_i\}$,
and swaps at distinct nodes act on disjoint qubit pairs, so their
time order does not change the branch state. The threshold holds
uniformly over all records, hence over every adaptive, sequential,
or parallel postselected protocol.

\paragraph{Heterogeneous links.}
The argument extends to link-dependent parameters. If the isotropic parameters are
$\eta_0,\ldots,\eta_N$, the filter-only and last-replacement weights
become
$$
w_*:=\prod_{i=0}^{N}\eta_i,
\qquad
w_k:=(1-\eta_k)\prod_{j=k+1}^{N}\eta_j,\qquad k=0,\ldots,N.
$$
They telescope to $\sum_{k=0}^{N}w_k=1-w_*$. The product-structure
and determinant-bound proofs are unchanged, since each depolarizing
step still preserves the determinant lower bound in
Eq.~\eqref{eq:depol-det-bound}. Therefore, $S\ge\tfrac12(1-w_*)$,
the coefficient of $\Omega_R$ is $w_*$, and the same argument proves
separability whenever $w_*=\prod_i\eta_i\le 1/3$. Applying the
reduction of Secs.~\ref{sm:slocc}--\ref{sm:caseA} separately to
each full-rank elementary link gives the heterogeneous extension
stated below Theorem~\ref{thm:rank4} in the main text.

\section{The entanglement budget: the route score and its low-noise approximation}\label{sm:budget}

Throughout, a route $P$ is a chain of $N+1$ entangled full-rank
elementary links $i=0,\dots,N$ with noise floors
$\lambda_i:=4\mu_{\min}(\rho_i)\in(0,1)$ and isotropic parameters
$\eta_i:=2(1-\lambda_i)/(2-\lambda_i)\in(0,1)$, as in
Eq.~\eqref{eq:eta-of-lambda}. Define
\begin{equation}\label{eq:chi-def}
\chi_i:=\ln\frac{1}{\eta_i}
=\ln\frac{2-\lambda_i}{2(1-\lambda_i)},
\qquad
\beta_i:=\frac{\chi_i}{\ln3}.
\end{equation}
The route score is
\begin{equation}\label{eq:beta-exact}
\beta(P):=\sum_{i\in P}\beta_i
=\frac{1}{\ln3}\sum_{i\in P}
\ln\frac{2-\lambda_i}{2(1-\lambda_i)},
\end{equation}
so $\beta_i$ is the link cost of main-text Eq.~\eqref{eq:beta-link}
and $\beta(P)$ is the route score of
main-text Eq.~\eqref{eq:budget}. Since
$\prod_{i\in P}\eta_i\le1/3$ if and only if
$\sum_{i\in P}\chi_i\ge\ln3$, that is, if and only if
$\beta(P)\ge1$, the heterogeneous form of
Theorem~\ref{thm:rank4}
(Sec.~\ref{sm:cancellation}) gives the branchwise implication
\begin{equation}\label{eq:beta-nogo}
\boxed{\;
\beta(P)\ge1
\ \Longrightarrow\
\substack{\text{every postselected}\\[-2pt]
\text{branch is separable}}\; .}
\end{equation}
The condition $\beta(P)<1$ means only that the
rejection rule of Eq.~\eqref{eq:beta-nogo} has not ruled out the route. It does not
guarantee entanglement delivery.

\begin{lemma}[Cost sandwich]\label{lem:chisandwich}
For every $\lambda\in(0,1)$,
$$
\frac{\lambda}{2}\ \le\ \chi(\lambda)\ \le\ \frac{\lambda}{2(1-\lambda)} .
$$
\end{lemma}

\begin{proof}
Write $\chi=\ln(1+x)$ with $x:=\lambda/(2(1-\lambda))$. The
elementary bounds $x/(1+x)\le\ln(1+x)\le x$ give the upper bound
directly and the lower bound
$\chi\ge x/(1+x)=\lambda/(2-\lambda)\ge\lambda/2$.
\end{proof}

\paragraph{Storage time folds into the link state.}
A qubit stored for a time $t$ in a memory with depolarizing time $T_q$
undergoes $\mathcal D_f(\varrho)=f\varrho+(1-f)\Tr(\varrho)\,I_2/2$, where
$f=e^{-t/T_q}$. In the postselected protocol, each qubit of the
chain is either measured in one swap or delivered to an end node, and
all of its storage decoherence precedes that single event. On every
realized branch, the end-to-end state is therefore the output of the
chain applied to the stored links
$$
\rho_i^{\mathrm{st}}:=(\mathcal D_{f_{A_i}}\otimes\mathcal D_{f_{B_i}})(\rho_i),
\qquad
f_{A_i}=e^{-t_{A_i}/T_{A_i}},\quad f_{B_i}=e^{-t_{B_i}/T_{B_i}},
$$
where $t_{A_i}$ and $t_{B_i}$ are the storage times of the two
qubits of link $i$, and $T_{A_i}$ and $T_{B_i}$ are the depolarizing
times of their memories. The memory
channel $\mathcal D_{f_{A_i}}\otimes\mathcal D_{f_{B_i}}$ is positive
and unital, so $\rho_i\succeq\mu_{\min}(\rho_i)I_4$ gives
$\rho_i^{\mathrm{st}}\succeq\mu_{\min}(\rho_i)I_4$. The stored link is
therefore again of full rank, with noise floor
$\lambda_i^{\mathrm{st}}:=4\mu_{\min}(\rho_i^{\mathrm{st}})\ge\lambda_i$. If
the stored link is separable, every branch is separable outright.
Otherwise, Theorem~\ref{thm:rank4} applies to the stored links, with
the noise floors $\lambda_i^{\mathrm{st}}$ evaluated at the storage
times realized on that branch. In general, $\rho_i^{\mathrm{st}}$
depends on $f_{A_i}$ and $f_{B_i}$ separately. For links with
maximally mixed marginals, it depends on them only through the
product $f_{A_i}f_{B_i}=e^{-\theta_i}$, where
\begin{equation}\label{eq:theta-def}
\theta_i:=\frac{t_{A_i}}{T_{A_i}}+\frac{t_{B_i}}{T_{B_i}}
\end{equation}
is the storage load of link $i$, and the stored noise floor has a
closed form.

\begin{lemma}[Noise-floor growth under storage]\label{lem:floorgrowth}
Let $\lambda_0:=4\mu_{\min}(\rho_i)$ be the initial noise floor of
the link $\rho_i$. If $\rho_i$ has maximally mixed marginals, then
the stored link is
$$
\rho_i^{\mathrm{st}}=e^{-\theta_i}\rho_i+(1-e^{-\theta_i})\,I_4/4 .
$$
The stored link, and hence its noise floor, depend on the storage
times only through $\theta_i$, and the stored noise floor, written
$\lambda_i(\theta_i):=\lambda_i^{\mathrm{st}}$, is
\begin{equation}\label{eq:lam-theta}
\lambda_i(\theta_i)=1-(1-\lambda_0)\,e^{-\theta_i}.
\end{equation}
\end{lemma}

\begin{proof}
Drop the index $i$. Write $\mathcal D_f=f\,\mathrm{id}+(1-f)\Rep$
and expand $\mathcal D_{f_A}\otimes\mathcal D_{f_B}$ into four
terms. When both marginals of $\rho$ are $I_2/2$, each of the three
terms that contain a replacement maps $\rho$ to $I_4/4$, so
$(\mathcal D_{f_A}\otimes\mathcal D_{f_B})(\rho)=f_Af_B\,\rho+(1-f_Af_B)\,I_4/4$,
which is the stated form of the stored link, since $f_Af_B=e^{-\theta}$.
Hence
$\mu_{\min}(\rho^{\mathrm{st}})=e^{-\theta}\mu_{\min}(\rho)+(1-e^{-\theta})/4$,
and with $\mu_{\min}(\rho)=\lambda_0/4$, multiplying by $4$ gives
$4\mu_{\min}(\rho^{\mathrm{st}})=\lambda_0e^{-\theta}+1-e^{-\theta}
=1-(1-\lambda_0)e^{-\theta}$, which is Eq.~\eqref{eq:lam-theta}.
\end{proof}

\paragraph{The low-noise score.}
Take $N+1$ identical links with initial noise floor $\lambda_0$ and
maximally mixed marginals. This covers the identical isotropic links
of the main text. The floors entering the score are the stored floors
$\lambda_i=\lambda_i(\theta_i)$ of Lemma~\ref{lem:floorgrowth}.
For small $\lambda_0$ and $\theta_i$, Eq.~\eqref{eq:lam-theta}
gives $\lambda_i=\lambda_0+\theta_i+O\bigl((\lambda_0+\theta_i)^2\bigr)$,
and Lemma~\ref{lem:chisandwich} gives
$2\chi_i=\lambda_i+O(\lambda_i^2)$, so
$\chi_i=(\lambda_0+\theta_i)/2+O\bigl((\lambda_0+\theta_i)^2\bigr)$.
Summing over the links in $\beta(P)=\sum_i\chi_i/\ln3$ gives
\begin{equation}\label{eq:budget-eng}
\beta(P)=\beta_{\mathrm{low}}
+O\Bigl(\sum_{i=0}^{N}(\lambda_0+\theta_i)^2\Bigr),
\qquad
\beta_{\mathrm{low}}:=
\frac{(N{+}1)\lambda_0+\sum_{i=0}^{N}\theta_i}{2\ln3},
\end{equation}
where $\sum_i\theta_i=\sum_q t_q/T_q$ runs over all $2(N+1)$
stored qubits. We call $\beta_{\mathrm{low}}$ the low-noise score.
When every memory has the same depolarizing time $T_m$, the leading
term is
\begin{equation}\label{eq:budget-low-Tm}
\beta(P)\approx\beta_{\mathrm{low}}
=\frac{(N{+}1)\lambda_0+\sum_q t_q/T_m}{2\ln3},
\end{equation}
which is the approximation
stated after Eq.~\eqref{eq:budget} in the main text. Replacing
$\beta(P)$ by $\beta_{\mathrm{low}}$ in Eq.~\eqref{eq:beta-nogo}, with
the loads $\theta_i$ realized on the branch, gives the approximate
rejection rule $\beta_{\mathrm{low}}\ge1$.

\paragraph{Accuracy of the low-noise score.}
The bounds below show that $\beta_{\mathrm{low}}$ approximates
$\beta(P)$ with an error, relative to $\beta(P)$, of first order in
$\lambda_0$ and $\theta_{\max}:=\max_i\theta_i$, and that
$\beta_{\mathrm{low}}\ge1/(1-\theta_{\max})$ already implies
$\beta(P)\ge1$.
Inserting $e^{-\theta_i}\le1-\theta_i+\theta_i^2/2$ and
$e^{-\theta_i}\ge1-\theta_i$ into Eq.~\eqref{eq:lam-theta} gives
\begin{align}
\lambda_i&\ \ge\ 1-(1-\lambda_0)\Bigl(1-\theta_i+\frac{\theta_i^2}{2}\Bigr)
=(\lambda_0+\theta_i)(1-\theta_i)+\frac{(1+\lambda_0)\theta_i^2}{2}
\ \ge\ (\lambda_0+\theta_i)(1-\theta_i),\notag\\
\lambda_i&\ \le\ 1-(1-\lambda_0)(1-\theta_i)
=\lambda_0+\theta_i-\lambda_0\theta_i\ \le\ \lambda_0+\theta_i .\notag
\end{align}
Lemma~\ref{lem:chisandwich} gives
$\lambda_i\le2\chi_i\le\lambda_i/(1-\lambda_i)$. The lower bound on
$\lambda_i$ gives the left inequality below. Since $\lambda/(1-\lambda)$
increases with $\lambda$, the upper bound $\lambda_i\le\lambda_0+\theta_i$
gives the right inequality, provided $\lambda_0+\theta_i<1$:
\begin{equation}\label{eq:2chi-sandwich}
(\lambda_0+\theta_i)(1-\theta_i)\ \le\ 2\chi_i\ \le\
\frac{\lambda_0+\theta_i}{1-\lambda_0-\theta_i}
\qquad(\lambda_0+\theta_i<1).
\end{equation}
Since $1-\theta_i\ge1-\theta_{\max}$ and, for
$\lambda_0+\theta_{\max}<1$, $1-\lambda_0-\theta_i\ge1-\lambda_0-\theta_{\max}>0$,
Eq.~\eqref{eq:2chi-sandwich} weakens to
$(\lambda_0+\theta_i)(1-\theta_{\max})\le2\chi_i\le
(\lambda_0+\theta_i)/(1-\lambda_0-\theta_{\max})$. Summing over $i$,
with $\sum_i(\lambda_0+\theta_i)=2\ln3\,\beta_{\mathrm{low}}$ and
$\sum_i2\chi_i=2\ln3\,\beta(P)$, and dividing by $2\ln3$ gives
$$
(1-\theta_{\max})\,\beta_{\mathrm{low}}
\ \le\ \beta(P)\ \le\
\frac{\beta_{\mathrm{low}}}{1-\lambda_0-\theta_{\max}} .
$$
Two conclusions follow. First, since $1-\lambda_0-\theta_{\max}\le1$,
the two bounds combine into
$$
|\beta(P)-\beta_{\mathrm{low}}|
\ \le\ \frac{\lambda_0+\theta_{\max}}{1-\lambda_0-\theta_{\max}}\,
\beta_{\mathrm{low}},
$$
and $\beta_{\mathrm{low}}\le\beta(P)/(1-\theta_{\max})$ by the lower
bound, so the error of $\beta_{\mathrm{low}}$ relative to $\beta(P)$ is
of first order in $\lambda_0$ and $\theta_{\max}$. Second, if
$\beta_{\mathrm{low}}\ge1/(1-\theta_{\max})$, then $\beta(P)\ge1$ by the
lower bound, and every postselected branch is separable by
Eq.~\eqref{eq:beta-nogo}, with no approximation involved.

\paragraph{Purification enters the score only through the noise floors.}
Purification of link $i$ is a finite-round SLOCC protocol between
its two holders acting on $m$ copies of $\rho_i$, and
Sec.~\ref{sm:slocc-reduction} shows that, because $\rho_i$ is of full
rank, on every accepted record it returns a pair $\rho_i'$ that is
separable or again of full rank. If
$\rho_i'$ is separable, every branch of the chain is separable
outright. Otherwise, $\rho_i'$ is an entangled full-rank link, and
Eqs.~\eqref{eq:beta-exact} and~\eqref{eq:beta-nogo} apply unchanged
with the purified noise floor $\lambda_i':=4\mu_{\min}(\rho_i')$ in
place of $\lambda_i$. Purification may lower $\lambda_i'$ below
$\lambda_i$ and so extend the horizon, but cannot remove it, which is
the purification statement below Theorem~\ref{thm:rank4} in the main
text. The score also fixes where purification can act. Applied
to a purely swapped segment of the route, Eq.~\eqref{eq:beta-nogo}
shows that the pair delivered by that segment is separable once the
segment's score reaches one, and a separable pair stays separable
under any SLOCC protocol. Purification after the swaps therefore
recovers nothing, and purification stations must be spaced so that,
at the least, every purely swapped segment between them has score
below one, as stated in the main text.

\section{Reduction from local protocols to pure swapping outcomes}
\label{sm:slocc-reduction}

We prove Corollary~\ref{cor:slocc-horizon} from the pure-outcome
bound already established, and we then prove the purification
statement of the main text, that a finite-round SLOCC protocol on
finitely many copies of a full-rank link returns a pair that is
separable or again of full rank. The reduction changes the allowed
protocol, not the threshold.

\paragraph{Model.}
The input is one copy of each link, that is, the chain state
$\rho_{\mathrm{ch}}:=\bigotimes_{i=0}^{N}\rho_i^{A_iB_i}$, which
reduces to $\rho_0^{\otimes(N+1)}$ for identical links. The left
boundary holds $A_0$, the right boundary holds $B_N$, and node $k$
holds $B_{k-1}A_k$ for $1\le k\le N$. Each party may add local
ancillas, apply local instruments in finitely many rounds,
communicate classically with every other party, and adapt later
instruments to earlier outcomes. We include shared randomness by
conditioning on its value and averaging at the end. No quantum
communication between parties is allowed. This model gives one copy
of each link, and we treat purification separately below. The protocol may accept any set of outcome records, and
all intermediate systems are traced out.

Set $\rho^{(t=0)}=\rho_{\mathrm{ch}}$. For the $i$th Kraus term of a
reported outcome in round $t=1$, let $F_i^{(t=1)}$, $E_{k,i}^{(t=1)}$,
and $G_i^{(t=1)}$ be the local operators at the left endpoint, node $k$,
and the right endpoint, respectively. Then
$$
\Theta_i^{(t=1)}=F_i^{(t=1)}\otimes E_{1,i}^{(t=1)}\otimes\cdots\otimes
E_{N,i}^{(t=1)}\otimes G_i^{(t=1)},
\qquad
\rho^{(t=1)}=\sum_i\Theta_i^{(t=1)}\rho^{(t=0)}
\Theta_i^{(t=1)\dagger}.
$$
Every round $t$ has Kraus operators of this product form because the
operations are local, and products of such operators have the same
form. Thus, every Kraus operator of the complete protocol is of the form
$$
\Theta_\alpha=F_\alpha\otimes E_{1,\alpha}\otimes\cdots\otimes
E_{N,\alpha}\otimes G_\alpha.
$$
Adaptivity changes the operators but not their product form. It
suffices to prove that, after tracing the intermediate systems, each
term $\Theta_\alpha\rho_{\mathrm{ch}}\Theta_\alpha^\dagger$ gives a
separable endpoint state.

For an orthonormal basis $\{\ket{r_k}\}_{r_k}$ of the system at node
$k$, write $\boldsymbol r=(r_1,\ldots,r_N)$ and
$\ket{\boldsymbol r}=\ket{r_1}\otimes\cdots\otimes
\ket{r_N}$. For each $\boldsymbol r$, define
$$
\widetilde\rho_{\alpha,\boldsymbol r}:=
\Tr_{\mathrm{int}}\!\left[
\left(\Id\otimes\ket{\boldsymbol r}\!\bra{\boldsymbol r}\otimes\Id\right)
\Theta_\alpha\rho_{\mathrm{ch}}\Theta_\alpha^\dagger
\right],
\qquad
\widetilde\rho_\alpha
=\sum_{r_1}\cdots\sum_{r_N}\widetilde\rho_{\alpha,\boldsymbol r}.
$$
Absorb the node operators into the basis vectors by defining
$\ket{w_{k,\alpha,r_k}}:=E_{k,\alpha}^\dagger\ket{r_k}$, while
$F_\alpha$ and $G_\alpha$ remain at the endpoints. Then
\begin{equation}
 \begin{aligned}
 \widetilde\rho_{\alpha,\boldsymbol r}
 &=(F_\alpha\otimes G_\alpha)
 \Bigl(\Id_{A_0}\otimes
  \bigotimes_{k=1}^{N}{}_{B_{k-1}A_k}\!\bra{w_{k,\alpha,r_k}}
  \otimes\Id_{B_N}\Bigr)\rho_{\mathrm{ch}}\\
 &\quad\times\Bigl(\Id_{A_0}\otimes
  \bigotimes_{k=1}^{N}\ket{w_{k,\alpha,r_k}}_{B_{k-1}A_k}
  \otimes\Id_{B_N}\Bigr)
  (F_\alpha^\dagger\otimes G_\alpha^\dagger).
 \end{aligned}
 \label{eq:slocc-sum}
\end{equation}
For each nonzero term, the operator between $F_\alpha\otimes G_\alpha$
and $F_\alpha^\dagger\otimes G_\alpha^\dagger$ has the form of
an entanglement-swapping chain with outcome vector
$\ket{w_{k,\alpha,r_k}}$ at every node $k$. If any outcome is product,
this operator is a product term. If all outcomes are entangled,
Theorem~\ref{thm:rank4} applies. Since $F_\alpha$ and $G_\alpha$ are
local, every $\widetilde\rho_{\alpha,\boldsymbol r}$, and therefore
$\widetilde\rho_\alpha$, is separable whenever
Eq.~\eqref{eq:horizon}, or its heterogeneous form, holds. Grouping
terms into any accepted record preserves separability.

\paragraph{Purification.}
Purification of a link is a finite-round SLOCC protocol between the two
holders of that link acting on $m$ copies of it. It is therefore the case
$N=0$ of the model above, with $\rho_{\mathrm{ch}}$ replaced by
$\rho_i^{\otimes m}$, which is positive definite because $\rho_i$ is of
full rank. There are no intermediate nodes to trace out, so the product
form established above reads $\Theta_\alpha=F_\alpha\otimes G_\alpha$,
where $F_\alpha$ and $G_\alpha$ map the $m$ qubits of a holder to its
output qubit, and the retained pair is
$$
\sigma=\sum_{\alpha}(F_\alpha\otimes G_\alpha)\,\rho_i^{\otimes m}\,
(F_\alpha^\dagger\otimes G_\alpha^\dagger),
$$
up to normalization, where the sum runs over the Kraus terms of the
accepted record. As matrices, $F_\alpha$ and $G_\alpha$ have two rows, one per
dimension of the output qubit, and $2^m$ columns. Their rank is at most
two, and rank two means full row rank, that is, surjectivity onto the
output qubit. For $m=1$, they are $2\times2$ matrices, and full row rank
means invertibility. Two cases remain.

\emph{Case 1: some $\alpha$ has
$\operatorname{rank}F_\alpha=\operatorname{rank}G_\alpha=2$.} Then
$F_\alpha\otimes G_\alpha$, a matrix with four rows and $4^m$ columns, has
full row rank, and $X\succ0$ implies $CXC^\dagger\succ0$ for every $C$ of
full row rank, so that term is positive definite on the two output qubits.
The remaining terms are positive semidefinite, so $\sigma$ has full rank.

\emph{Case 2: every $\alpha$ has a factor of rank at most one.} A factor of
rank one has the form $F_\alpha=\ket f\!\bra g$, where $\bra g$ acts on the
$m$ qubits held by one party and $\ket f$ is a vector of that party's
output qubit. Such a factor turns the term into
$\ket f\!\bra f\otimes G_\alpha(\bra g\otimes\Id)\rho_i^{\otimes m}
(\ket g\otimes\Id)G_\alpha^\dagger$, a product operator, and likewise when
$G_\alpha$ has rank one. Every term is then a product operator, so $\sigma$
is separable.

Purification therefore returns a separable pair or again a full-rank pair
$\rho_i'$. A chain of such pairs obeys Theorem~\ref{thm:rank4} with the
noise floors $\lambda_i=4\mu_{\min}(\rho_i')$. Purification may reduce
these floors and extend the horizon, but never removes it.

\subsection*{Part B: Chains of low-rank links can evade the horizon}

\section{\texorpdfstring{Low-rank closure under $\Phi^\pm$ swaps}{Low-rank closure under Phi+/- swaps}}
\label{sm:lowrank}

For $s\in(0,1]$ and $\kappa\in(0,1]$, consider the amplitude-damping and
pure-dephasing channel
$\mathcal E_{s,\kappa}:=\mathcal P_\kappa\circ\mathcal A_s$ of
Example~\ref{ex:lowrank}, where $\mathcal A_s$ is amplitude
damping in which the excited state survives with probability $s$, and
$\mathcal P_\kappa(X):=\frac{1+\kappa}{2}X+
\frac{1-\kappa}{2}\sigma_zX\sigma_z$ is pure dephasing with coherence
factor $\kappa$. Sending one qubit of $\PhiP$ through
$\mathcal E_{s,\kappa}$, we have, in the computational basis, the
elementary link
\begin{equation}\label{eq:lowrank-link}
\rho_{s,\kappa}:=(\II\otimes\mathcal E_{s,\kappa})(\PhiP)
=\frac12\begin{pmatrix}
1 & 0 & 0 & \kappa\sqrt s\\
0 & 0 & 0 & 0\\
0 & 0 & 1-s & 0\\
\kappa\sqrt s & 0 & 0 & s
\end{pmatrix}.
\end{equation}%
This is main-text Eq.~\eqref{eq:ad-link}. At $\kappa=1$, the
pure dephasing channel $\mathcal P_\kappa$ acts as the identity, so
amplitude damping alone produces the link. At $s=1$, the amplitude
damping channel $\mathcal A_s$ acts as the identity, so pure
dephasing alone produces the link. At $s=\kappa=1$, the link is
$\PhiP$.

\begin{theorem}[Low-rank persistence under postselected swapping]
\label{thm:lowrank-persistence}
The rank of $\rho_{s,\kappa}$ is three for $0<s,\kappa<1$, two when
only one of the two channels acts nontrivially, and one at
$s=\kappa=1$. The state is entangled throughout $s,\kappa\in(0,1]$. Chain $N+1$
links $\rho_{s_i,\kappa_i}$ of the family, not necessarily identical,
and retain either the $\Phi^+$ or $\Phi^-$ outcome at every
intermediate Bell measurement, applying a local Pauli-$Z$ correction
for each $\Phi^-$ outcome. The normalized end state, total accepted
probability, and concurrence are
\begin{equation}\label{eq:lowrank-chain}
\rho_N=\rho_{S,K},
\quad
S:=\prod_{i=1}^{N+1}s_i,
\quad
K:=\prod_{i=1}^{N+1}\kappa_i,
\quad
p_N^{(\Phi^\pm)}=2^{-N},
\quad
C\!\left(\rho_N\right)=K\sqrt S>0.
\end{equation}%
For $N+1$ identical copies, this reads
$\rho_N=\rho_{s^{N+1},\kappa^{N+1}}$ with concurrence
$(\kappa\sqrt s)^{N+1}$. Since $S=1$ requires every $s_i=1$ and $K=1$
requires every $\kappa_i=1$, a chain whose links all see the same
channels returns a link seeing those same channels, of the same rank.
Hence, chains of rank-three and of rank-two links both deliver
an entangled state at every finite chain depth.
\end{theorem}

Direct calculation from Eq.~\eqref{eq:lowrank-link} gives the ranks
stated in the theorem and, on $\operatorname{span}\{\ket{01},\ket{10}\}$,
the partial-transpose block
\begin{equation}\label{eq:lowrank-ppt}
\frac12\begin{pmatrix}
0 & \kappa\sqrt s\\
\kappa\sqrt s & 1-s
\end{pmatrix},
\qquad
\det=-\frac{\kappa^2s}{4}<0,
\end{equation}%
with negative determinant, so $\rho_{s,\kappa}$ is entangled, and the
$X$-state concurrence formula gives $C(\rho_{s,\kappa})=\kappa\sqrt s$.
For two links $\rho_{s_1,\kappa_1}^{AB}$ and $\rho_{s_2,\kappa_2}^{CD}$,
direct calculation of the contraction over the measured qubits $B$ and
$C$ gives
\begin{align}
\widetilde\rho_\pm^{AD}
&:={}_{BC}\!\bra{\Phi^\pm}
\bigl(\rho_{s_1,\kappa_1}^{AB}\otimes
      \rho_{s_2,\kappa_2}^{CD}\bigr)
\ket{\Phi^\pm}_{BC}
=\frac18\begin{pmatrix}
1 & 0 & 0 & \pm\kappa_1\kappa_2\sqrt{s_1s_2}\\
0 & 0 & 0 & 0\\
0 & 0 & 1-s_1s_2 & 0\\
\pm\kappa_1\kappa_2\sqrt{s_1s_2} & 0 & 0 & s_1s_2
\end{pmatrix}\notag\\
&\phantom{:}=\frac14\,Z_\pm
\rho_{s_1s_2,\kappa_1\kappa_2}^{AD}Z_\pm,
\qquad
Z_+=I_4,\quad Z_-=\sigma_z\otimes I_2.
\label{eq:lowrank-swap-law}
\end{align}%
The two $\Psi^\pm$ outcomes populate $\ket{01}$ and so leave the
family, which is why we retain only $\Phi^\pm$. Both retained
outcomes have probability $1/4$, and after the indicated correction
one accepted swap multiplies both $s$ and $\kappa$. Iterating over
the $N$ swaps gives the state and concurrence in
Eq.~\eqref{eq:lowrank-chain}, and accepting two of the four Bell
outcomes at each node gives the total probability $2^{-N}$.

\subsection*{Part C: Experimental Procedure, Uncertainty Analysis, \& Device Calibration Data}

\section{Experimental Procedure}\label{sm:repswaps}
Any 2-qubit density matrix can be reconstructed through an eigen-decomposition approach. The states used in this approach can be built using quantum circuits. To reconstruct the density matrix, we can mesure the circuits independently and classically reconstruct the system.

$$
\rho = \sum_{i} \lambda_i \ket{\psi_i}\!\bra{\psi_i}
\rightarrow \begin{quantikz}[row sep={0.7cm,between origins},column sep=0.3cm]
\lstick{$q_0$} & \qw   & \gate[wires=2]{\ket{\psi_i}}   & \qw         & \meter{}\\
\lstick{$q_1$} & \qw      &\qw &  \qw     & \meter{}\\
\end{quantikz} \rightarrow \rho = \sum_i \lambda_i \rho_i
$$

Assume $\rho_{AD}$ represents the density matrix of One-Sided Amplitude Damping with parameter $p_{AD}$ and $\rho_{DP}$ represents the density matrix of One-Sided Depolarizing channel with parameter $p_{DP}$. To ensure a fair comparison, we calibrate the initial density matrices to have equal concurrence $C_{\mathrm{target}}$ by modifying their respective parameters. Thus, all

$$
\begin{aligned}
p_{\mathrm{AD}}
&\longrightarrow
\rho_{\mathrm{AD}}\!\left(p_{\mathrm{AD}}\right)
\longrightarrow
C_{\mathrm{AD}}
=
C\!\left(\rho_{\mathrm{AD}}\right)
\approx C_{\mathrm{target}},
\\[0.5em]
p_{\mathrm{DP}}
&\longrightarrow
\rho_{\mathrm{DP}}\!\left(p_{\mathrm{DP}}\right)
\longrightarrow
C_{\mathrm{DP}}
=
C\!\left(\rho_{\mathrm{DP}}\right)
\approx C_{\mathrm{target}}.
\end{aligned}
$$

To bypass an exponentially decaying post-selection success, entanglement swapping is restricted within a 4 qubit system. At each swapping round, $q_0q_1$ is prepared as the 2-qubit system from the previous round, while $q_2q_3$ is prepared as a link of the same channel time with concurrence $C_\mathrm{target}$. Bell-State Measurement (BSM) is then performed and we post-select the $\ket{00}$ outcome on $q_1q_2$, entangling $q_0q3$. Lastly, tomography is performed on $q_0q_3$ to reconstruct the 2-qubit density matrix, which is then eigen-decomposed and used as $q_0q_1$ in the next round of swapping.

$$
\begin{tikzpicture}[
    node distance=0.85cm,
    box/.style={
        draw,
        rounded corners,
        align=center,
        minimum width=5.2cm,
        minimum height=0.85cm
    },
    arrow/.style={-{Stealth[length=2mm]}, thick}
]
    \node[box] (input) {
        Previous state $\rho^{(n)}_{01}$
        \quad+\quad
        fresh matched link $\rho^{(0)}_{23}$
    };

    \node[box, below=of input] (bsm) {
        Bell-state measurement on $q_1q_2$
        and post-selection on $00$
    };

    \node[box, below=of bsm] (tomo) {
        Tomography of outer pair $q_0q_3$
    };

    \node[box, below=of tomo] (reconstruct) {
        Reconstruct swapped state $\rho^{(n+1)}_{03}$
    };

    \node[box, below=of reconstruct] (eigen) {
        Eigendecompose $\rho^{(n+1)}$ and prepare its
        pure-state components for the next round
    };

    \draw[arrow] (input) -- (bsm);
    \draw[arrow] (bsm) -- (tomo);
    \draw[arrow] (tomo) -- (reconstruct);
    \draw[arrow] (reconstruct) -- (eigen);

    \draw[arrow]
        (eigen.west)
        -- ++(-1.2,0)
        |-
        node[pos=0.25, left, align=left]
        {\footnotesize use as $\rho^{(n+1)}_{01}$}
        (input.west);
\end{tikzpicture}
$$

\section{Uncertainty Analysis}\label{sm:uncertainty}
Uncertainties in the concurrence values of entanglement swapping were estimated using a resampling of the measured tomography counts. For each tomography circuit, count datasets with the same number of shots were sampled from the experimentally observed outcome distribution recovered from experiment runs. Each dataset underwent the same post-selection, density-matrix reconstruction, and concurrence calculation as the original data.

For $N$ samples with concurrence values $C_n$, the uncertainty reported as an error bar is the sample standard deviation,
$$
    \sigma_C =
    \sqrt{
        \frac{1}{N-1}
        \sum_{n=1}^{N}
        \left(C_n-\overline{C}\right)^2
    },
    \qquad
    \overline{C} = \frac{1}{N}\sum_{n=1}^{N} C_n.
$$

\section{Device Calibration Data}\label{sm:qubits}
Single-qubit and two-qubit calibration parameters for \texttt{ibm\_boston} last updated on 21 July 2026 at 8:57:55 UTC, used for the entanglement-swapping experiments.

\bigskip
\centering
\begin{tabular}{@{}cccc@{}}
\toprule
Qubit & $T_1$ (\si{\micro\second}) &
$T_2$ (\si{\micro\second}) & Measurement error \\
\midrule
86 & \num{204.24} & \num{268.56} & $2.319 \times 10^{-3}$ \\
87 & \num{273.31} & \num{281.93} & $3.052 \times 10^{-3}$ \\
88 & \num{301.02} & \num{384.76} & $2.93 \times 10^{-3}$ \\
89 & \num{273.18} & \num{357.23} & $2.075 \times 10^{-3}$ \\
\bottomrule

\end{tabular}

\vspace{0.7em}

\begin{tabular}{@{}cccc@{}}
\toprule
Qubit 1 & Qubit 2 & CZ Error & RZZ error \\
\midrule
86 & 87 & $1.711 \times 10^{-3}$ & $1.51 \times 10^{-3}$ \\
87 & 88 & $1.114 \times 10^{-3}$ & $1.134 \times 10^{-3}$ \\
88 & 89 & $1.142 \times 10^{-3}$ & $9.544 \times 10^{-4}$ \\
\bottomrule
\end{tabular}

\end{document}